\documentclass[sigconf]{acmart}

\copyrightyear{2026}
\acmYear{2026}
\setcopyright{cc}
\setcctype{by}
\acmConference[KDD '26]{Proceedings of the 32nd ACM SIGKDD Conference on Knowledge Discovery and Data Mining V.2}{August 09--13, 2026}{Jeju Island, Republic of Korea}
\acmBooktitle{Proceedings of the 32nd ACM SIGKDD Conference on Knowledge Discovery and Data Mining V.2 (KDD '26), August 09--13, 2026, Jeju Island, Republic of Korea}
\acmDOI{10.1145/3770855.3818145}
\acmISBN{979-8-4007-2259-2/2026/08}

\usepackage[most]{tcolorbox}
\usepackage{caption}
\usepackage{balance}
\usepackage{graphicx}
\usepackage{textcomp}
\usepackage{optidef} %

\usepackage{algorithm}
\usepackage{enumitem}
\usepackage{subcaption}
\usepackage[noend]{algpseudocode}
\usepackage[subtle]{savetrees}
\usepackage{amsmath,amsfonts}
\usepackage{amsthm}

\newtheorem{theorem}{Theorem}
\newtheorem{lemma}[theorem]{Lemma}

\theoremstyle{plain}
\newtheorem{proposition}[theorem]{Proposition}

\theoremstyle{definition}
\newtheorem{definition}[theorem]{Definition}
\newtheorem{example}[theorem]{Example}
\theoremstyle{remark}
\newtheorem{remark}[theorem]{Remark}

\usepackage[table]{xcolor}
\usepackage{multirow}

\newcommand{\multigroup}{\texttt{$g$-GBLR}}
\newcommand{\arbitrary}{\texttt{SGBLR}}

\newcommand{\singAlg}{\texttt{ERMB}}
\newcommand{\multAlg}{\texttt{MGRA}}
\newcommand{\maxsat}{\textsc{Max-2-SAT}}
\newcommand{\maxunisat}{\textsc{Max-1-in-2-SAT}}
\newcommand{\maxcut}{\textsc{MaxCut}}
\newcommand{\comphs}{comparison hyperplane}
\newcommand{\Comphs}{Comparison hyperplane}
\newcommand{\ilpbase}{$\mathtt{MILP_{Base}}$}
\newcommand{\ilprefined}{$\mathtt{MILP_{Refined}}$}
\newcommand{\sampling}{\texttt{LocalSearch}}

\newtcolorbox{examplebox}[1][]{%
  breakable,
  colback=blue!10,
  colframe=blue!20,
  arc=2mm, %
  fonttitle=\bfseries,
  boxrule=0mm,
  boxsep=1mm,
  left=0mm,
  right=0mm,
  top=0mm,
  bottom=0mm,
  enhanced jigsaw,
  #1
}

\begin{document}

\title{Explaining Rankings with Hidden Group Bonuses}

\author{Alvin Hong Yao Yan}
\orcid{0009-0007-6146-5871}
\affiliation{%
  \institution{National University of Singapore}
  \city{Singapore}
  \country{Singapore}
}
\email{alviny@u.nus.edu}

\author{Suraj Shetiya}
\orcid{0000-0001-9166-2365}
\affiliation{%
  \institution{Indian Institute of Technology Bombay}
  \city{Mumbai}
  \country{India}
}
\email{surajs@cse.iitb.ac.in}

\author{Sujoy Bhore}
\orcid{0000-0003-0104-1659}
\affiliation{%
  \institution{Indian Institute of Technology Bombay}
  \city{Mumbai}
  \country{India}
}
\email{sujoy@cse.iitb.ac.in}

\author{Priyanka Golia}
\orcid{0009-0004-0704-226X}
\affiliation{%
  \institution{Indian Institute of Technology Delhi}
  \city{Delhi}
  \country{India}
}
\email{pgolia@cse.iitd.ac.in}

\author{Diptarka Chakraborty}
\orcid{0000-0002-0687-5823}
\affiliation{%
  \institution{National University of Singapore}
  \city{Singapore}
  \country{Singapore}
}
\email{diptarka@nus.edu.sg}

\renewcommand{\shortauthors}{Alvin Hong Yao Yan, Suraj Shetiya, Sujoy Bhore, Priyanka Golia, and Diptarka Chakraborty}

\begin{abstract}
Determining a linear utility function that correlates with observed candidate rankings is a foundational problem with applications in domains such as admissions, hiring, and recommendation systems, e.g., [Storandt and Funke, AAAI'19, Zhang et al., KDD'23, Wang et al., ICDE'24 (best paper award), Chen and Wong, VLDB'24]. Traditionally, these models assume full visibility into the feature sets used to determine the utility score. However, real-world scenarios often involve sensitive attributes that are hidden or partially observed, yet may influence outcomes through additive bonuses designed to promote fairness, as in [Gale and Marian, ICDE'24]. Motivated by such practical concerns, we study a variant of the ranking explanation problem where sensitive features are unobserved but may influence candidate rankings through group-specific linear boosts.

We present a formal framework for modeling this problem and develop an algorithmic solution that leverages constraint satisfaction and automated reasoning techniques to jointly infer the linear scoring parameters and latent group bonuses consistent with the observed rankings. We further show that determining a satisfying linear function with group-specific bonuses is \textsf{NP}-hard in general, but when the feature dimension and the number of groups are constant, the problem admits a polynomial-time solution. Our approach is the first to address this nuanced variant, which captures key real-world challenges in fair ranking and admission systems. We perform extensive experiments on both real-world and synthetic datasets, demonstrating that our method effectively recovers hidden bonus structures and provides faithful explanations of observed ranking outcomes.
\end{abstract}



\maketitle

\newcommand\kddavailabilityurl{https://doi.org/10.5281/zenodo.20322450}
\ifdefempty{\kddavailabilityurl}{}{
\begingroup\small\noindent\raggedright\textbf{Resource Availability:}\\
The source code associated with this paper is publicly available at \url{\kddavailabilityurl}.
\endgroup
}

\section{Introduction}\label{sec:intro}

Ranking algorithms are central to numerous AI applications, including web search, recommender systems, and information retrieval.
However, despite their pervasive use, the decision-making processes behind these algorithms often remain opaque, which raises critical concerns around trust, fairness, and accountability. As a result, explainability in ranking has emerged as a pressing research challenge within the AI community. 
Several recent works have sought to adapt principles from interpretable machine learning to the ranking setting~\cite{chowdhury2023rank, chowdhury2025rankshap,heuss2025rankingshap}, aiming to generate faithful, user-understandable explanations for ranked outputs. 
Some approaches focus on post-hoc explanations, while others design models that are easier to interpret from the start~\cite{singh2019policy, ai2018unbiased}. 
There is also increasing attention to broader goals like fairness and trust in ranked results~\cite{doshi2017towards,rampisela2024can}. Yet, many basic questions remain 
open: \emph{what does it mean for a ranking to be explainable? How do we measure the quality of an explanation when the output is an ordered list instead of a single label?}

Modeling user preferences is a central challenge in multi-criteria decision-making. A common and effective approach is to formalize these preferences with a \emph{utility function} that assigns a numerical score/utility to each candidate, thereby inducing a natural ordering. Candidates are typically characterized by a set of features, which can be distinguished as \emph{scoring attributes} -- that directly contribute to utility, and \emph{sensitive attributes} -- used to enforce auxiliary goals such as fairness or diversity. Among various functional forms, linear utility models are arguably the most widely adopted. Their prevalence stems from a powerful combination of simplicity and practical efficacy. The straightforward structure of a weighted sum not only facilitates deployment in large-scale systems like web search and recommender engines~\cite{morik2020controlling,guo2023rankdnn,thonet2022listwise} but also offers inherent transparency~\cite{sen2020curious,gale2020explaining}. This transparency, simplicity, and effectiveness make linear utility-based models a natural starting point for research into explainability and learning to rank~\cite{asudeh2019rrr, chen2024robust, chen2025synthesizing}, interactive learning~\cite{nanongkai2012interactive, wang2021interactive, zhang2023finding}, reverse regret~\cite{wang2024reverse} (best paper award at ICDE'24), and the integration of fairness and diversity mechanisms~\cite{liu2024fair, guo2025towards, wang2025interactive}. Despite their simplicity, linear utilities have demonstrated competitive performance, confirming their status as a robust and versatile modeling choice, e.g.,~\cite{qian2015learning, storandt2019algorithms}. We refer the readers to Section~\ref{sec:related} for a more elaborate discussion on related works.

Several works have studied how to quantify the distance between a given ranking and one induced by a (linear) utility function. Common notions include the number of violated pairwise orders, the distance to a consistent utility vector, and the minimum perturbation needed to make the ranking ``linearizable''. Exact computation is often intractable in the presence of conflicting preferences, motivating the study of approximation and structural relaxations (e.g.,~\cite{dwork2001rank,davenport2004computational}). 

In many real-life scenarios, a small additive bonus is essential to ensure that a linear utility-based ranking model aligns with observed outcomes. Such a bonus is primarily calculated based on the sensitive attributes~\cite{shetiya2022fairness, asudeh2019designing, wang2025interactive} -- separate from the scoring attributes used to compute utility -- associated with each candidate or item. This often arises when external goals, such as affirmative action, fairness constraints, or prior trust scores, must be incorporated uniformly across items of certain groups. The application of compensatory additive bonuses to linear utilities is an approach that has been explored for promoting fairness in ranking systems~\cite{gale2024explainable}. Without such a term, even simple adjustments to promote underrepresented groups may not be achievable within the linear feature space, making the additive bonus a minimal yet powerful mechanism for enforcing policy-driven interventions (see~\cite{mathioudakis2020affirmative}).

We motivate our problem with the following running example of diversity in University admissions.
\begin{table}[t]
\centering
\small
\setlength{\tabcolsep}{8pt}
\rowcolors{2}{gray!8}{white}
\begin{tabular}{@{}lcccc@{}}
\toprule
\multirow{2}{*}{\textbf{Candidate}} & \multirow{2}{*}{\textbf{Group}} & \textbf{Test} & \textbf{SAT}  & \multirow{2}{*}{\textbf{Rank}}\\
\hiderowcolors %
& & \textbf{Score} & \textbf{Score} & \\
\showrowcolors %
\midrule
$c_1$ & - & 9.8 & 2.0 & 3 \\
$c_2$ & - & 8.1 & 7.8 &  1 \\
$c_3$ & - & 7.2 & 8.6 & 2 \\
$c_4$ & - & 6.9 & 4.2 & 7\\
$c_5$ & UP & 6.0 & 3.2 & 4 \\
$c_6$ & UP & 3.7 & 7.1 & 5 \\
$c_7$ & - & 5.1 & 8.0 & 6\\
$c_8$ & UP & 4.5 & 3.5 & 8\\
\bottomrule
\end{tabular}
\caption{Example dataset of eight university applicants.}
\label{table:running-example-dei}
\vspace{-10mm}
\end{table}

\begin{examplebox}\label{exmpl:university}
    \begin{example}[University admissions]
        Consider a public university conducting its annual admissions, where applicants are ranked using a \emph{linear utility-based model}. Academic performance is assessed through multiple indicators, such as entrance test and SAT scores, which are normalized, weighted, and linearly aggregated to compute a base utility score.

        As a public institution committed to diversity, equity, and inclusion (DEI), the university augments this model with an \emph{additive incentive mechanism}. Recognizing that applicants from underprivileged (UP) backgrounds may face structural disadvantages affecting academic performance, an expert committee assigns a group-specific bonus score. This bonus is added to the base utility score of each applicant in the corresponding group, yielding an adjusted utility score that determines the final ranking.

        For example, suppose eight applicants apply, as shown in Table~\ref{table:running-example-dei}. The candidates are ranked $c_2$, $c_3$, $c_1$, $c_5$, $c_6$, $c_7$, $c_4, c_8$ based on their adjusted utility scores. Given only anonymized public data (academic indicators and rankings), third-party auditors aim to explain the observed outcomes under this \emph{linear utility function with additive group-level bonuses}. The sensitive attributes, such as group membership (Column 2 of Table~\ref{table:running-example-dei}), and the exact bonus values are not disclosed.

    \end{example}
\end{examplebox}

\paragraph{Problem (Informal) Formulation} Given a ranking, can it be explained by a linear utility function with a small number of additive bonuses? In the \textsc{Singleton Group Bonus Linear Ranking} (\arbitrary) version, arbitrary bonuses may be given to at most $k$ items. In the \textsc{$g$-Group Bonus Linear Ranking} (\multigroup) version, there are $g$ (special) groups each containing a limited number of items, and group-specific bonuses -- each item of a particular group gets the same amount as a bonus, and bonus value for different groups could be different -- are allowed. The goal is to find weights and bonus assignments that reproduce the ranking exactly.

Note that Example~\ref{exmpl:university} is for the \texttt{$1$-GBLR} variant with one (special) group being UP -- all the candidates of the UP group receive an additive bonus of the same amount. %

Our method conceptually builds upon seminal works LIME~\cite{ribeiro2016should} and SHAP~\cite{10.5555/3295222.3295230}, which infer feature contributions to predictions, as our model learns weights for each feature, which provides an explanation for the observed ranking. It extends this line of work by modeling hidden additive bonuses beyond observable features as part of the explanation. We envision our technique as contributing to the broader area of trustworthy and responsible data science, particularly in applications such as algorithmic auditing and explainability. For example, in an auditing setting, suppose a ranking model is used to evaluate candidates, and there is a requirement that all candidates within the same group receive identical bonuses. If our method is unable to recover a linear utility function with such additive bonuses that explains the observed ranking, this serves as evidence that the ranking violates the intended fairness constraint. Alternatively, our framework can be used for interpretability. If an observed ranking can be well-approximated by our model, the learned weights provide insight into the relative importance of different attributes, while the inferred additive bonuses reveal how much advantage may have been assigned to specific candidates. This can help assess whether such adjustments are reasonable within the given domain and context.

\subsection{Our Contributions}
In this paper, we initiate a formal study on the problem of learning a linear utility function with small additive group-wise bonuses that realizes a given ranking. On the theoretical side, we provide exact algorithms and establish computational hardness. Next, we address the question of designing scalable algorithms that work efficiently in practice. More specifically, we show the following results.
\begin{itemize}[leftmargin=*]
    \item We start by considering the singleton group variant (\arbitrary) that captures the allowance of arbitrary (potentially different) additive bonuses to a few individual items. We design a hyperplane-based algorithm -- enumerating over all the regions specified by a set of comparison hyperplanes, and then performing certain \emph{longest common subsequence (LCS)} computation in each region. Our algorithm works in polynomial time for constant-dimensional points. We further extend our algorithm to a more general variant of group-wise bonuses (\multigroup) by introducing a new auxiliary dimension per group, leading to an algorithm that runs in polynomial time for a constant number of groups and for constant-dimensional points. Note that often in practice, both the number of groups and the dimension of the feature vectors are indeed small (bounded by some constant).
    \item We complement our algorithmic results with computational hardness. The problem, even if in its simplest form of singleton groups, is \texttt{NP}-hard for arbitrary dimensions, establishing that the exponential dependency on dimension of our algorithm is unavoidable assuming \texttt{P} $\ne$ \texttt{NP}.
    \item While the theoretical algorithm establishes tractability in principle at least for constant dimension, its high complexity motivates us to explore more practical alternatives using modern ILP solvers. We formulate the problem as an MILP and evaluate its performance using state-of-the-art solvers. 
    \item We conduct extensive experiments on both real-world and synthetic datasets to demonstrate the efficacy and scalability of our methods. In case of \arbitrary\, we observe that for two-dimensional data the MILP formulations scale well, whereas the computational hardness inherent in the hyperplane-based algorithm prevents it from scaling beyond small instances.  In the case of \multigroup\ with $g = 1$, the MILP formulations exhibit strong scalability with respect to both dimensionality (up to $d = 17$) and the number of tuples (up to approximately $75{,}000$). For the real-world dataset, containing around $300{,}000$ tuples, the MILP formulations successfully handle the \multigroup\ problem for both $g = 1$ and $g = 2$, highlighting their practical applicability at large scale.
\end{itemize}    

\section{Preliminaries}\label{sec:prelim}
In this section, we briefly introduce the notations used in the paper and then formally describe the problem formulation.

\noindent\textbf{Data model}
Consider a dataset $\mathcal{D} = \{t_1, t_2, \ldots, t_n\}$ of $n$ tuples, each with $d$ numerical attributes $\mathcal{A} = A_1, A_2, \ldots, A_d$ and additional categorical attributes (such as \emph{gender} and \emph{race}), some of which are sensitive and not revealed during the ranking process. We use $t_i[j]$ to denote the value of tuple $t_i$ for attribute $A_j$.

\noindent\textbf{Utility model} As part of the black-box scoring process,  each tuple $t_i\in\mathcal{D}$ is scored based on its values for the numerical attributes $\mathcal{A}$. While there are numerous ``\emph{models}'' that have been used to score tuples (such as \emph{monotone}~\cite{fagin2001optimal}, \emph{linear}~\cite{chang2000onion,das2006answering}), we adapt the widely used linear utility functions to this work, as in several previous works in the literature, e.g.,~\cite{wang2024reverse, chen2024robust, wang2025interactive}.
\emph{Linear utility function} $s_w$ consists of a $d$ dimensional weight vector $w=\{w_1, w_2, \ldots, w_d\}$. Given a linear function $s_w$, the score/utility for tuple $t$ is computed as $s_w(t)=\sum_{j=1}^d w_j \cdot t[j]$.
As each attribute in the dataset has a non-zero contribution towards scoring, $\forall_{j\in [d]}\ w_j\neq 0$.

Without loss of generality, we assume that $\forall_{j\in [d]}\ w_j \ge 0$ and $\sum_{j \in [d]} w_j = 1$; such assumptions are justified in~\cite{nanongkai2012interactive,  wang2021interactive}. For the sake of completeness, we provide a brief discussion on positive weights in Appendix~\ref{sec:app-positiveweight}.

\noindent\textbf{Bonus-aware utility function}: Based on certain attributes which are termed as \emph{sensitive attributes}, bonus scores may be allotted to items which belong to certain groups. 

In our running example, under-privileged candidates in the dataset would receive an additive bonus based on historical biases. We model this as an additive bonus on top of the linear utility function score. All tuples sharing the same sensitive attribute $\tau$ receive the same bonus, denoted by \texttt{bonus}($t[\tau]$). 

Based on this formulation, the bonus-aware utility function, which accounts for the sensitive attributes, is given by,
\vspace{-1mm}
\begin{align}
    f_w(t)=\texttt{bonus}(t[\tau])+\sum\nolimits_{j=1}^d w_j \cdot t[j]
\end{align}
\vspace{-4mm}

For simplicity, we denote the bonus-based utility function $f_w$ as $f$ throughout the paper.

To illustrate, in Example~\ref{exmpl:university}, the bonus-based utility function for the tuples in the under-privileged group \texttt{UP} is given by $f_w(t)=2\times test + 1\times SAT + 5$, where $5$ is the additive bonus given only to \texttt{UP} candidates.

\noindent\textbf{Ranking}: The scores of the tuples using the bonus-based linear utility function are used to rank the tuples. As a consequence of the ranking process, a complete ranking over the $n$ input tuples can be observed. We refer to this ranking 
as $\pi$.

\noindent\textbf{Realizable by linear functions}: 
Given a dataset $\mathcal{D}$ and a ranking $\pi$, we say that $\pi$ is \emph{realizable} if there exists a linear utility function $s_w$ such that applying $s_w$ to $\mathcal{D}$ produces the ranking $\pi$. Importantly, this function $s_w$ does not involve any additive bonuses. We denote such a function as \texttt{s}\textsuperscript{-1}($\pi$).

If no such realizable linear function exists, it indicates that the observed ranking $\pi$ must have been generated using the black-box ranking process that includes some form of additive bonus. In this case, we aim to explain $\pi$ by constructing a linear function augmented with additive group-based bonuses. As part of this process, the dataset is partitioned into ``\emph{groups}'' and a fixed bonus is assigned to each group, such that a linear function with these bonuses can realize $\pi$.

Without any constraint on the number of groups, one can trivially assign each tuple to its own ``\emph{under-represented}'' group and learn arbitrary bonuses. To avoid such degenerate solutions, we restrict the number of groups to at most $g$.

\noindent\textbf{Realizable by bonus-aware linear functions}: Given a dataset $\mathcal{D}$, a ranking $\pi$, and a budget of $g+1$ groups, we say that $\pi$ is \emph{realizable by a bonus-aware linear function} if there exists a partition of the tuples into $g+1$ groups, where the majority group receives a bonus of $0$, such that a linear function $f_w$, combined with these group bonuses, produces the ranking $\pi$. We denote such a function as $f^{-1}_g(\pi)$.

We are now ready to formally introduce the two problems that form the focus of this work.

\paragraph{Problem formulation}

For our first problem setting, we consider finding a bonus-aware linear function that realizes a given ranking $\pi$ where the number of tuples that receive a bonus is at most $k$.

\begin{definition}[\arbitrary]
    Given a set of tuples $\mathcal{D} = \{t_1, t_2, ... t_ n \}$ where each $t_i \in \mathbb{R}^d$, a target ranking $\pi$ over the tuples and an integer $k$, does there exist a linear utility function $s_w$ with $w \in \mathbb{R}^d$ and a bonus vector $\textbf{v} \in \mathbb{R}^n$ with at most $k$ non-zero values such that sorting the tuples by $s_w(t_i) + v_i$ produces the ranking $\pi$? 
\end{definition}

In practice, an estimate of the number of tuples belonging to under-represented groups is often available. The \arbitrary\ problem captures the lower bound on the number of tuples that need to be allotted a bonus for a linear function to realize $\pi$. 

The number of under-represented groups is typically a small known number (usually constant), and thus it would be of interest to find a bonus-based linear utility function that respects the number of groups $g$. We now propose the \multigroup\ problem.

\begin{definition}[\multigroup]
    Given a set of tuples $\mathcal{D} = \{t_1, t_2, ... t_ n \}$ where each $t_i \in \mathbb{R}^d$, a ranking $\pi$ over $\mathcal{D}$ and integer $k$. Does there exist a weight vector $w \in \mathbb{R}^d$, bonus values $v_1, \ldots, v_g \in \mathbb{R}^+$, and disjoint subsets $G_1, \ldots, G_g \subseteq \mathcal{D}$ where $\sum_{i \in [g]} |G_i| \le k$, such that ranking $\pi$ can be realized by $s_w$ using $v_i$ as bonus values for tuples in $G_i$?
\end{definition}

Note, some tuples in $\mathcal{D}$ may not belong to $\cup_{i \in [g]} G_i$.

\section{Algorithmic Framework}

\subsection{Algorithmic Results for \arbitrary}
We start by considering the \arbitrary\ problem and provide an algorithm to solve this problem.

\begin{theorem}
\label{theorem:singletongroup-weight}
    There exists a deterministic algorithm that, given a set of tuples $\mathcal{D} = \{t_1, t_2, ... t_ n \}$ where each $t_i \in \mathbb{R}^d$, a ranking $\pi$ over $\mathcal{D}$ and an integer $k$, if there exists a satisfying solution to \arbitrary, returns a linear utility function $s_w$ with $w \in \mathbb{R}^d$ and a bonus vector $\textbf{v} \in \mathbb{R}^n$ with at most $k$ non-zero values such that sorting the tuples by $s_w(t_i) + v_i$ produces the ranking $\pi$, in time $\mathcal{O}(n^{2d+1} \log n)$.
\end{theorem}

We devote the rest of the subsection to proving the above theorem. Before describing our approach, we introduce the notion of ``\emph{\comphs}'' and then propose algorithms which explore regions of the weight-space based on \emph{\comphs}.

\begin{definition}[\Comphs] Consider two tuples $t_i\in \mathbb{R}^d$ and $t_j\in \mathbb{R}^d$. The set of all weight vectors $\mathcal{C}_{i,j}$ where the tuple $t_i$ has a score equal to $t_j$ is defined as a \comphs. Formally,  $\sum_{a=1}^d w_a\cdot t_i[a]=\sum_{a=1}^d w_a\cdot t_j[a]$.
\end{definition}

Given two tuples $t_i$ and $t_j$, a \comphs\ divides the $d$-dimensional weight space into two parts such that for any weight vector in $\mathcal{C}_{i,j}^+$, $t_i$ is better than $t_j$, i.e., 

\vspace{-3mm}
\begin{equation}\label{eqn:cij-plus}
    \mathcal{C}_{i,j}^+=\{w\in \mathbb{R}^d\ |\ \sum\nolimits_{a=1}^d w_a \cdot t_i[a]>\sum\nolimits_{a=1}^d w_a \cdot t_j[a]\},
\end{equation}
\vspace{-3mm}

\noindent and for any weight vector on the other side $\mathcal{C}_{i,j}^-$, $t_j$ is ranked better than $t_i$.

Our idea stems from a simple observation that, while there are $n!$ orderings of $n$ tuples, in $d$-dimensional space, not all of these are realizable through linear utility functions. Instead, we will soon see that only $n^{2d-2}$ rankings are realizable. We then exploit this property to design an efficient algorithm for both problems.

\paragraph{Algorithm for {\arbitrary}}
We first present our approach, and then provide an argument about its correctness. Consider the $d$-dimensional weight-space. For every pair of tuples $t_i$, $t_j$ introduce the \comphs\ $\mathcal{C}_{i,j}$ in this space. We term the set of these $\binom{n}{2}$ hyperplanes as $\mathcal{H}$. Introduction of $\mathcal{H}$ to the $\mathbb{R}^d$ dimensional weight-space partitions it into numerous regions, which is termed the arrangement of hyperplanes~\cite{edelsbrunner1987algorithms}. As each weight-function in this space lies either below or above the \emph{\comphs s}, the ranking that is realized corresponds to a unique ranking. 

The set of all functions that have the same ranking in the arrangement can be shown to be a convex region and corresponds to a ``\emph{cell}" of the arrangement. To formalize this, consider a vector of $\binom{n}{2}$ boolean signs corresponding to each utility function termed as the \emph{location vector}. A location vector is of the form $\{+,-\}^{\binom{n}{2}}$. For example, the $(i,j)$-th location contains the sign $-$ or $+$ which represents the location of the weight vector with respect to $\mathcal{C}_{i,j}$, \emph{i.e.}, if the sign is $+$, then the tuple $t_i$ is scored better than $t_j$. A similar argument follows for $-$. The set of ranking functions with the same location vector can be found by the intersection of $\binom{n}{2}$ half-spaces; the set of ranking functions is convex. Also, note that this corresponds by definition to the cell of an arrangement~\cite{edelsbrunner1987algorithms}.

Based on this, our approach to solve \arbitrary\ is to enumerate the regions in space formed by $\mathcal{H}$, and in turn, enumerate the realizable rankings, $\mathcal{R}$. For every ranking $\rho\in \mathcal{R}$, we compute the least number of tuples that need to be allotted a bonus to obtain the ranking $\pi$, which is equivalent to finding a \emph{Longest Common Subsequence (LCS)} between $\pi$ and $\rho$. Since both the inputs to the LCS problems are rankings, it is well-known that the problem is essentially the same as finding a \emph{Longest Increasing Subsequence (LIS)} in $\rho$ where the reference ranking is $\pi$ -- by renaming the tuples in ascending order based on $\pi$, \emph{i.e.}, we give the tuple located at $\pi[0]$ a new number of $t_0$ and so forth. 

The tuples that are not part of the LIS need to be allotted a bonus so that the linear ranking function $s^{-1}(\rho)$ can realize $\pi$. Over the set of rankings $\rho\in \mathcal{R}$, our algorithm outputs the linear function that had the longest value for the increasing subsequence. The tuples that are not part of the subsequence need to be added with a bonus to obtain $\pi$.  We present the pseudo-code for our approach in \texttt{\singAlg} (Algorithm~\ref{algo:singleton-groups}).

\begin{algorithm}[t]

\caption{\texttt{\singAlg} : \textsc{Enumerate Rankings and Minimize Bonuses} }\label{algo:singleton-groups}
\begin{algorithmic}[1]
\small
\Procedure{ERMB}{Dataset $\mathcal{D}$, target ranking $\sigma$}
\State Initialize $\mathcal{H} \gets \emptyset$
\ForAll{pairs $(t_i, t_j)$ where $1 \le i < j \le n$}
    \State Add comparison hyperplane $\mathcal{C}_{i,j}$ to $\mathcal{H}$
\EndFor
\State Enumerate regions formed by the arrangement of $\mathcal{H}$
\State Let $\mathcal{R} \gets$ all rankings $\rho$ realizable in each region of $\mathcal{H}$
\State Initialize $\rho^* \gets \texttt{null}$,\ $\text{maxLCS} \gets 0$,\ $\mathcal{G}\gets \texttt{null}$
\ForAll{rankings $\rho \in \mathcal{R}$}
    \State Compute $L \gets \text{LCS}(\rho, \pi)$
    \If{$|L| > \text{maxLCS}$}
        \State $\text{maxLCS} \gets |L|$; \ \ 
        \State $\rho^* \gets \rho$;\ \ 
        \State $\mathcal{G}\gets [n]-L$;
    \EndIf
\EndFor
\State \Return linear function $s^{-1}(\rho^*)$, $\mathcal{G}$
\EndProcedure
\end{algorithmic}
\end{algorithm}

\paragraph{Analysis of \texttt{\singAlg}} 
The analysis of our algorithm relies on two steps: (i) enumeration of the cells of the arrangement, and (ii) LIS of the sequences $\rho$. To prove the complexity of the enumeration step,  
we rely on 
Corollary 28.1.2 of \cite{toth2017handbook}. The corollary states that the number of such regions with $n$ hyperplanes in $d$-dimensional space is bounded by $\mathcal{O}(n^{d})$. We state the Corollary below for easy access.

\begin{proposition}~\cite{toth2017handbook}
    The maximum combinatorial complexity of an arrangement of $n$ hyperplanes in $\mathbb{R}^d$ is $\mathcal{O}(n^d)$. Moreover, if the arrangement is simple, its complexity is indeed $\Theta(n^d)$. %
\end{proposition}

Observe that each hyperplane $C_{i,j}$ as shown in Equation~\ref{eqn:cij-plus} has an intercept of $0$ and hence passes through the origin. Therefore, the resulting arrangement is a special case of a hyperplane arrangement known as a central arrangement (see \cite{schneider2022convex} for details). In a central arrangement, the combinatorial complexity of a central arrangement with $m$ hyperplanes in $d$ dimensions is $\mathcal{O}(m^{d-1})$ (see \cite{schneider2022convex} Chapter~5). Thus, in our case, $n^2$ hyperplanes in $ d$-dimensional weight space form a central arrangement, which bounds the number of unique rankings realizable to $\mathcal{O}(n^{2(d-1)})$. This helps us in establishing $\mathcal{O}(n^{2d+1} \log n)$ bound on the time complexity (see Appendix~\ref{sec:app-ermb}). %

Next, we show the correctness of our algorithm \texttt{\singAlg} (Algorithm~\ref{algo:singleton-groups}), which will complete the proof of Theorem~\ref{theorem:singletongroup-weight}.

\begin{lemma}
    \label{theorem:arbitraryk-weights}
    \texttt{\singAlg} (Algorithm~\ref{algo:singleton-groups}) returns a linear utility function with weight vector $w \in \mathbb{R}^d$ which solves the \arbitrary\ problem.
\end{lemma}
\begin{proof}
    Suppose vectors $w^*$, $v^*$ represent the optimal weight vector and optimal bonus values, respectively. The linear function $f_{w^*}$ with bonus of $v^*$ orders $t_i \in \mathcal{D}$ by decreasing value such that the ordering matches $\pi$ and $v^*$ is the most sparse among all such solutions.

    Now, consider the region of the arrangement of $\mathcal{C}_{i,j}$ where the vector $w^*$ lies in. Let $w$ be the vector in this region that our algorithm picks, and let $\rho$ be the resulting ranking from ordering $t \in \mathcal{D}$ in decreasing order of their scores. For any pair of points $t_1, t_2 \in \mathcal{D}$ that receive no bonus from $v^*$, it is clear that $w \cdot t_1 > w \cdot t_2$ if and only if $w^* \cdot p_1 > w^* \cdot t_2$. Therefore, both $\rho$ and $\pi$ must have the same total order over points in $\mathcal{D}$ that receive no bonus from $v^*$. These points are a lower bound on the length of the LCS of $\pi$ and $\rho$, so our approach must find a $w$, $v$ which satisfies \arbitrary\ with $v$ at least as sparse as $v^*$.
\end{proof}

\noindent \textbf{Improved Running Time for $2$-dimensions.} We remark that in the special case of $d = 2$, the enumeration step (line 4 of Algorithm~\ref{algo:singleton-groups}) can be performed very efficiently as the weight space is now $1$-dimensional. Therefore, each comparison hyperplane is a point partitioning the line. Enumeration of the regions can thus be done in $\mathcal{O}(n^2 \log n)$ time by sorting the comparison hyperplanes, followed by a linear scan to obtain the ranking and then performing a LIS on each ranking, leading to an overall running time of $\mathcal{O}(n^3 \log n)$ (compared to an worse bound of $\mathcal{O}(n^5 \log n)$ implied from Theorem~\ref{theorem:singletongroup-weight}). We use this faster implementation in our experiments.

\subsection{Algorithm for {\multigroup}}
The budgeted group setting is more intriguing as Algorithm~\ref{algo:singleton-groups} cannot directly be adapted to this setting, because the notion of a fixed number of groups is not modeled in \arbitrary. In this section, we first embed the idea of \emph{bonus} into the weight space and then design a polynomial-time algorithm for the fixed dimension, fixed group setting. 

We treat group bonuses as \emph{bonus dimensions}, expanding the feature space to $d+g$ dimensions -- one additional dimension per group. Then, to handle group membership, the algorithm generates $g+1$ copies of each tuple: one base version and one for each group membership (whether it belongs to a particular group or not). It then uses hyperplane arrangements in this augmented space to enumerate all possible rankings of the expanded dataset. Unlike simpler cases that use Longest Increasing Subsequence, this approach requires us to employ a budgeted Longest Common Subsequence (LCS). A dynamic programming routine aligns the target ranking $\pi$ with these candidate rankings to find the optimal solution that strictly respects the group size constraints. We provide a detailed description of the algorithm and analysis in Appendix~\ref{sec:app-gblr}.

\section{Hardness Results}
In this section, we show that the problem of linear realizability of a ranking with additive group bonuses is \texttt{NP}-hard.

\begin{theorem}\label{thm:np-hard}
{\arbitrary} is \texttt{NP}-hard.
\end{theorem}

To show the above, we provide a reduction from a variant of the classical {\maxsat}, namely {\maxunisat}. 

\begin{definition}[\maxunisat]
    Given a 2-CNF formula, and a non-negative integer $r$, the problem is to decide whether there exists an assignment of variables that satisfies at least $r$ clauses by making \emph{exactly one} literal (out of two literals) in each of these $r$ clauses to be true.
\end{definition}

The above problem is known to be \texttt{NP}-hard by a straightforward reduction from {\maxcut} -- another classical \texttt{NP}-hard problem. In fact, the problem remains \texttt{NP}-hard even for monotone formulas.

\begin{theorem}[Folklore]\label{thm:np-hard-sat}
{\maxunisat} is \texttt{NP}-hard.
\end{theorem}

\noindent \textbf{Reduction from {\maxunisat}. }Suppose we are given an instance of the {\maxunisat} problem: A 2-CNF formula $F$ over $n$ variables $x_1,\cdots, x_n$, with $m$ clauses $C_1,C_2,\cdots, C_m$ (without loss of generality assume that $m < n^2$), and a non-negative integer $r$. Then we construct an instance of the decision version of the {\arbitrary} problem as follows: For each clause $C_i$, create a point $p_i \in \mathbb{R}^n$ by setting, %

\vspace{-3mm}
\[
p_i[j] = \begin{cases}
    1 \quad \quad \quad \text{if }x_j \text{ appears in }C_i\\
    -1 \quad\quad \; \text{if }\bar{x}_j \text{ appears in }C_i\\
    0 \quad \quad \quad \text{otherwise}.
\end{cases}
\]
\vspace{-1mm}

Further, take $\ell = (m+1)n^2$ additional points: $q_1=\cdots=q_{\ell}=0^n$.
Consider the dataset $\mathcal{D}:= \cup_{i=1}^m \{p_i\} \bigcup \cup_{i=1}^{\ell} \{q_i\}$. Then, consider the following ranking
$\pi:= q_1 \preceq \cdots \preceq q_{n^2} \preceq p_1 \preceq q_{n^2+1} \preceq \cdots \preceq q_{2n^2} \preceq p_2 \preceq \cdots \preceq q_{\ell - 2n^2 + 1} \preceq \cdots \preceq q_{\ell - n^2} \preceq p_m \preceq q_{\ell - n^2 + 1} \preceq \cdots \preceq q_{\ell}$.

The instance of the decision version of the {\arbitrary} problem would constitute $\mathcal{D}$, ranking $\pi$, and the non-negative integer $m-r$. We defer the soundness and completeness proof to Appendix~\ref{sec:app-NP-hard}.

\begin{remark}
    It is worth noting that {\arbitrary} is a special case of the more general {\multigroup} problem, where each group could be of arbitrary size. Thus \texttt{NP}-hardness of {\arbitrary} immediately establishes a similar hardness for {\multigroup}. Furthermore, one may note that our reduction works even when there are only two (non-singleton) groups with bonuses.
\end{remark}

\section{Scalable Frameworks via MILP} \label{sec:ilp}

The geometric algorithm introduced earlier provides an exact, constructive way of identifying valid rankings through hyperplane arrangements. However, such geometric reasoning scales poorly with the number of data points $n$ and dimensions $d$, since the number of combinatorial regions grows exponentially. Even for moderate $d$, enumerating or traversing these regions becomes practically intractable. To design a more scalable alternative, we reformulate {\multigroup} problem as a \textbf{Mixed Integer Linear Program (MILP)} — a mature framework that allows us to express logical ordering and grouping constraints in a linear-algebraic form, solvable by powerful off-the-shelf engines.

The MILP formulation represents ranking relationships through a set of linear inequalities defined over two classes of variables — continuous feature weights $w$ and binary membership indicators $\delta_{ir}$. Each inequality enforces the required order between tuples in terms of their weighted feature scores and group bonuses. This formulation connects combinatorial ranking constraints with continuous optimization, allowing discrete ordering relations to be reasoned about using standard MILP solvers. 

\subsection{{\ilpbase} Formulation}
We first define the base formulation, denoted as \textbf{{\ilpbase}}, which represents the canonical encoding of the ranking constraints:

\vspace{-3mm}
\begin{equation*}
\begin{array}{lll}
\text{minimize} & 0 & \\
\text{subject to} 
& \sum_{j \in [d]} (t_i[j] - t_{(i+1)}[j]) \cdot w_j  + \sum_{r \in [g]} \delta_{ir} \cdot v_r & \\
& - \sum_{r \in [g]} \delta_{(i+1)r} \cdot v_r \ge 0 & \forall i < n  \\
& \sum_{r \in [g]} \delta_{ir} \le 1 & \forall i \in [n]  \\
& \sum_{i \in [n]} \sum_{r \in [g]} & \delta_{ir} \le k \\
& \delta_{ir} \in \{0, 1\},  v_r \ge 0 \quad \quad \quad \quad \quad \quad\ \ \ \ \ \ \ \forall i \in [n],\ &r \in [g] \\
& |w_j| \geq 1 & \forall j \in [d] 
\end{array}
\end{equation*}
\vspace{-1mm}

Here, $w \in \mathbb{R}^d$ denotes the feature weights to be learned, $v_r$ denotes the additive group bonus for group $G_r$. We introduce \( g \cdot n \) binary indicator variables \(\delta_{ir}\) for \( i \in [n] \) and \( r \in [g] \), where
$\delta_{ir}$ is 1 if tuple $t_i$ belongs to group $G_r$, otherwise $0$. The constraints encode ordering, group capacity, uniqueness, non-negativity, and weight-magnitude requirements. The objective is to minimize 0 (i.e., this is a feasibility problem), which means the ILP seeks any feasible solution satisfying all constraints that respects the ranking $\pi$.

Notice that the {\ilpbase} formulation remains generic and domain-agnostic: it treats all features and groups symmetrically, without leveraging problem-specific insights. In many real-world settings, additional knowledge—such as bounds on group influence, or on pre-identified boosted tuples—can meaningfully restrict the feasible space and guide the solver toward more computationally efficient solutions. Incorporating such application-driven structure motivates a refined formulation, {\ilprefined}, described next.

\subsection{{\ilprefined} Formulation}
\textbf{{\ilprefined}} -- The refined formulation tightens the feasible region by incorporating domain/application-informed assumptions and introducing solver-oriented constraints that eliminate redundant solutions and enhance numerical stability.

{\ilprefined} introduces the following enhancements:
\begin{enumerate}[leftmargin=*]
    \item \textbf{Strict ordering with tolerance:}  
    To ensure a clear and stable ranking, we introduce a small positive margin~$\epsilon$ between consecutive tuples. This enforces a strictly decreasing order of adjusted scores rather than allowing equality, thereby avoiding ties and improving numerical stability during optimization. We adopt a strategy similar to that of prior work~\cite{chen2025synthesizing}. In our experiments, we set $\epsilon$ as $10^{-5}$, which ensured sufficient spacing between tuple scores, thus accommodating floating-point tolerances. 
    
    \item \textbf{Bounded additive bonuses:}  
    Each group’s additive bonus is constrained within a predefined range to prevent it from dominating the linear utility function. These bounds capture realistic limitations on how much a group-specific adjustment can influence the overall ranking. We emphasize this bound is not a restriction of {\ilprefined}, and can be changed based on expectations of the range using domain knowledge.

    \item \textbf{Non-negative feature weights:}   In many real-world settings, such as performance evaluation or credit scoring, features contribute positively to the final score. Accordingly, all feature weights are constrained to be non-negative, reflecting domain semantics and eliminating redundant sign symmetries in the optimization process. We emphasize that this assumption is not a limitation of {\ilprefined}; if a feature is known to contribute negatively to the final utility, one can appropriately scale the feature values to ensure this.

    \item \textbf{Dominated tuples:} Under the assumption that feature weights are non-negative, we use the standard idea of dominance followed in monotone functions~\cite{borzsony2001skyline}. If it is observed that there exists a pair of tuples where $t_j$ dominates $t_i$, yet $t_i$ is ranked before $t_j$, then it is guaranteed that $t_i$ must have received an additive bonus. Therefore, we add the constraint that $t_i$ must belong to one of the groups receiving a bonus by adding a constraint $\sum_{r \in [g]} \delta_{ir} = 1$. 
    \vspace{0.5mm}
    
    Finding the set of points that dominate at least one other point yet are ranked lower can be formulated as a skyline query~\cite{borzsony2001skyline}. To this end, the dataset is augmented with an additional dimension representing the rank of each tuple. Specifically, for a tuple $t_i$, we append its rank in $\pi$ as the $(d+1)$\emph{th} dimension. On the resulting augmented dataset, we compute the skyline—where larger values are preferred in every dimension—using the Double Divide-and-Conquer (\texttt{DDC}) algorithm, which runs in $\mathcal{O}(n\log^{d-1} n + nd)$ time~\cite{
    kung1975finding}.  Any tuple that is dominated by another must necessarily receive a bonus. 
    This procedure can be repeated iteratively until no additional dominated tuples are identified.

\end{enumerate}

The complete {\ilprefined} model is thus:

\vspace{-3mm}
\begin{equation*}
\begin{array}{lll}
\text{minimize}  & 0 & \\
\text{subject to} 
& \sum_{j \in [d]} (t_i[j] - t_{(i+1)}[j]) \cdot w_j  + \sum_{r \in [g]} \delta_{ir} \cdot v_r \\
& - \sum_{r \in [g]} \delta_{(i+1)r} \cdot v_r \ge \epsilon & \forall i < n\\
& \sum_{r \in [g]} \delta_{ir} \le 1 & \forall i \in [n] \\
& \sum_{i \in [n]} \sum_{r \in [g]} \delta_{ir} \le k & \\
& 0 \leq v_r \leq 100 & \forall r \in [g] \\
& w_j \ge 0  &\forall j \in [d]
\end{array}
\end{equation*}
\vspace{-2mm}

\paragraph*{Relationship Between {\ilpbase} and {\ilprefined}}

The relationship between the two formulations is hierarchical. The feasible region of {\ilprefined} is a strict subset of that of {\ilpbase}, ensuring that every refined solution remains sound under the base model. While {\ilpbase} defines a broad feasibility space, {\ilprefined} incorporates additional structure—such as bounded bonuses and restricted membership subsets—that enhances interpretability, aligns with domain semantics, and yields more realistic solutions. In essence, {\ilpbase} serves as the foundational declarative model, whereas {\ilprefined} is its solver-optimized and practically scalable counterpart.

\section{Experimental Evaluation}\label{sec:exp}
We now present an empirical evaluation of our proposed formulations and algorithms. The primary goals of this experimental evaluation are to assess scalability and practical applicability, and to compare the proposed methods with one another and with suitable baselines. We examined both real-world and synthetic datasets to understand how the models behave under diverse data distributions and parameter settings. 

\subsection{Experimental setup}

\paragraph{Experimental environment.} All experiments were run on a server with an AMD EPYC 7763 CPU, with 128GB of RAM allocated for each experiment, running Ubuntu 24.04. All code is implemented using Python 3.12.3 with Gurobi version 12.0.2 as the MILP solver. Gurobi and implemented algorithms were allowed to use up to 16 threads. The timeout was considered 1800 seconds.

\smallskip
\noindent\textbf{Datasets.} We evaluate our methods on one real-world dataset and a set of synthetically generated datasets.

\smallskip
\noindent\emph{Real-world data.} 
The real-world dataset, introduced in prior work~\cite{cai2025findingfairscoringfunction}, contains examination scores of 384,977 candidates who took the Joint Entrance Examination (JEE) in India in 2009. After removing entries with missing values, 384,970 candidates remain, each described by three numerical attributes—Mathematics, Physics, and Chemistry scores—and categorical attributes such as gender and reservation category. We use the three numerical scores as equally weighted features. Categorical attributes define groups receiving additive bonuses. In single-group experiments, female candidates receive a bonus of 4, reflecting the median score difference from male candidates. In two-group experiments, candidates from underprivileged categories receive a bonus of 6, and female candidates receive a further bonus of 4, corresponding to their median differences from the general (GE) category. These data-driven adjustments preserve the original score distribution while modeling group-level advantages.

\paragraph*{Synthetic data} We also construct multiple synthetic datasets to systematically evaluate scalability. Each dataset \( \mathcal{D} = \{t_1, t_2, \ldots, t_n\} \) consists of \( n \)~tuples, where each tuple \( t_i \) is a \( d \) dimensional feature vector.  
Each feature component \( t_i[j] \) is sampled independently from a uniform distribution over \([0,25)\) and rounded to two decimal places. A linear weight vector \( W = (w_1, \ldots, w_d) \) is sampled similarly.  
To introduce structured deviations from a purely linear ranking, additive bonuses are assigned to randomly generated groups of tuples. For each group $G_r$, a bonus \( v_r \) is sampled uniformly from \([5d,10d)\).  For {\ilprefined}, a random subset \( \mathcal{D}' \subseteq \mathcal{D} \) is first selected, and groups are then formed by uniformly sampling tuples from \( \mathcal{D}' \).  Each tuple’s augmented score is computed as $f(t_i) = \sum_{j=1}^d w_j t_i[j] + v_r $
if it belongs to $G_r$, and as $ f(t_i) = \sum_{j=1}^d w_j t_i[j] $ otherwise.  Finally, tuples are sorted in decreasing order of their augmented scores to yield the target ranking~\( \pi \).  This generation process models ranking scenarios in which a base linear utility function is modified by sparse, group-level adjustments, providing a controlled ground truth for evaluation.

\paragraph{\textbf{Baselines}} 
We consider three baselines for comparison. We use ordinal regression~\cite{srinivasan1976linear} as the first baseline, where ranks are used as the ordinal predictor values. As a second baseline, we apply logistic regression by creating a dataset of all pairs of tuples. For a pair $t_i$ ranked above $t_j$, the tuple $t_i - t_j$ is labeled $1$, and $t_j - t_i$ labeled $0$. This results in a supervised learning task with binary labels, on which standard logistic regression can be applied. Both of these baselines output a linear utility function, and given this function, the number of tuples that need to be given a bonus to realize the input function is computed.
We also consider a simple baseline, called {\sampling}, which repeatedly samples a random unit vector as the weight function within the given time limit. For each sampled weight vector, it computes the minimum number of tuples that must receive an additive bonus to realize the target ranking~$\pi$. The best-performing weight vector -- i.e., the one minimizing the number of required bonuses -- is returned as the final output. 

The code is available at the following link \footnote{\url{https://github.com/Trustworthy-Ranking-Data-Management-lab/explaining-rankings-hidden-bonus}}. 
Our experimental evaluation examines both the computational performance and the solution quality of the proposed formulations and algorithms.

\subsection{Research Questions}

\begin{figure*}[t]
    \centering
    \begin{subfigure}[t]{0.33\textwidth}
        \includegraphics[width=\linewidth,height=3.5cm]{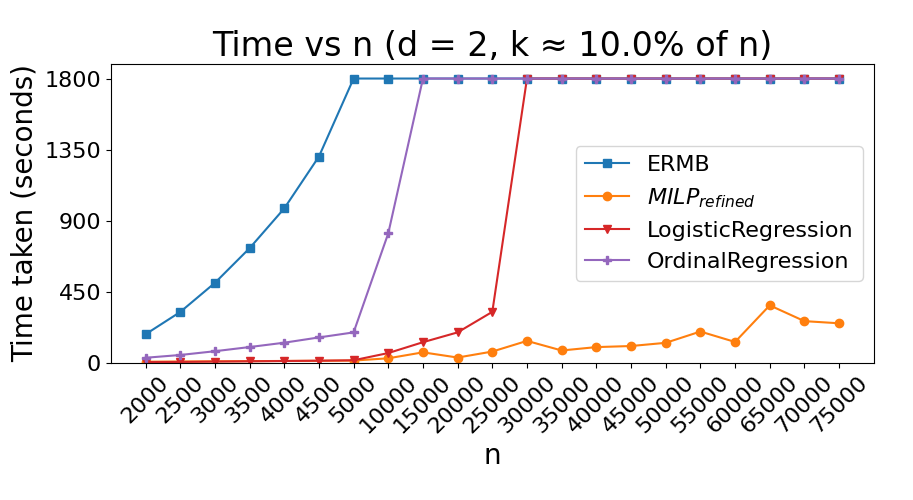}
        \vspace{-7mm}
        \caption{\footnotesize [Singleton groups setting] Runtime analysis of ERMB and \ilprefined on 2D instances.}
        \label{fig:singleton-time-d2}
    \end{subfigure}
    \hfill
    \begin{subfigure}[t]{0.27\textwidth}
        \includegraphics[width=\linewidth,height=3.2cm]{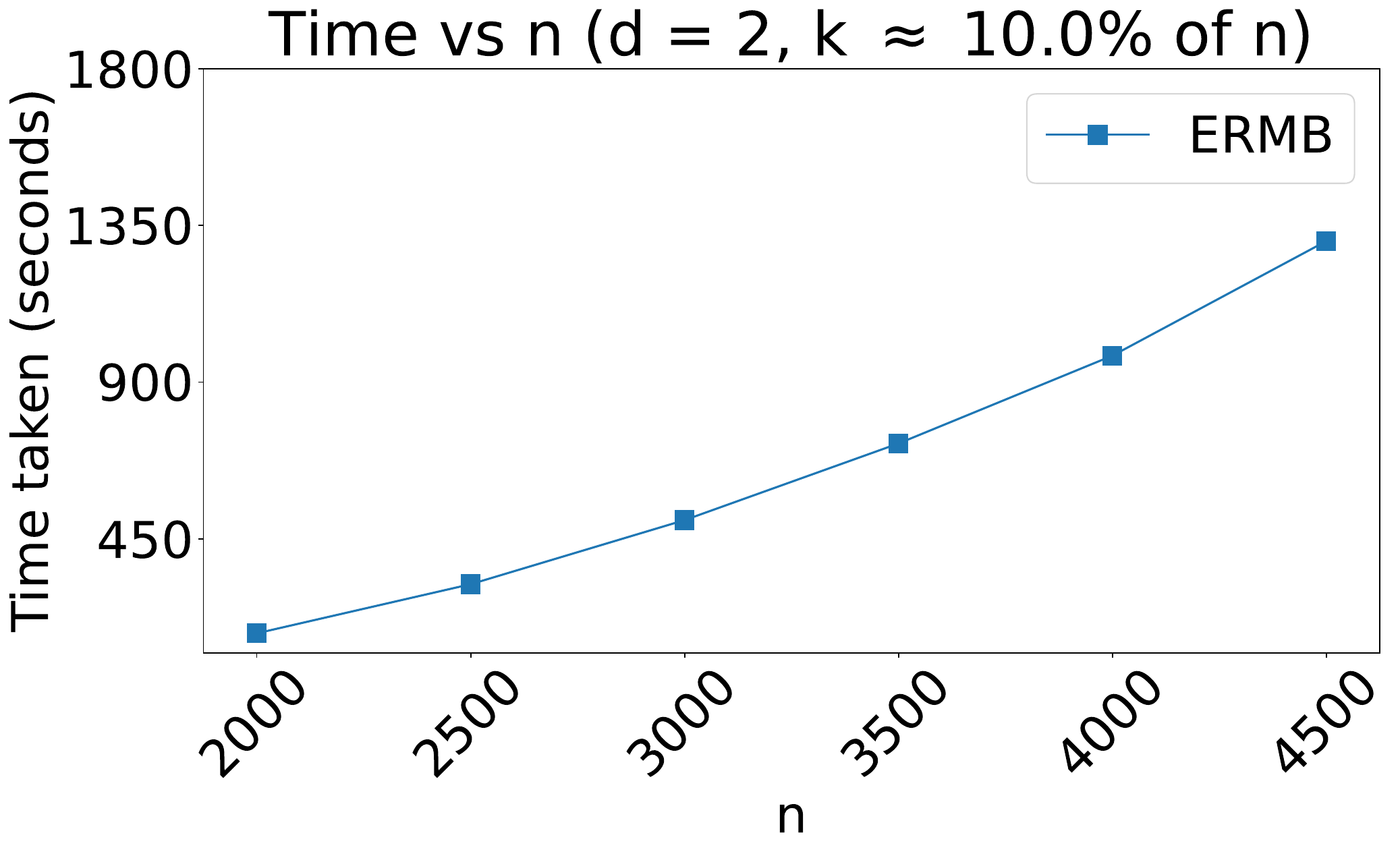}
        \vspace{-7mm}
        \caption{\footnotesize [Singleton groups setting]  Runtime analysis of ERMB on 2D instances.}
        \label{fig:singleton-ermb}
    \end{subfigure}
    \hfill
    \begin{subfigure}[t]{0.33\textwidth}
    \centering
    \includegraphics[width=\linewidth,height=3.5cm]{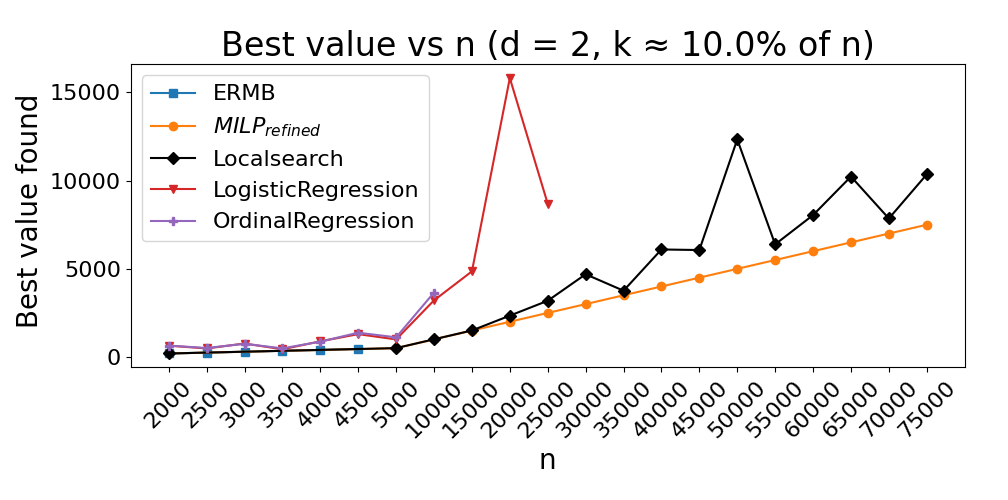}
    \vspace{-7mm}
    \caption{\footnotesize [Singleton groups setting]  Quality analysis for ERMB, \ilprefined and {\sampling} on 2D instances.}
    \label{fig:singleton-best}
    \end{subfigure}
    \vspace{-3mm}
     \caption{Analysis of ERMB}
     \vspace{-3mm}
\end{figure*}

\subparagraph*{\textbf{Singleton Setting.}}
\begin{enumerate}[label={RQ\arabic*:}]
    \item How do the algorithms perform in the singleton-groups setting, and how does their computational behavior vary with increasing instance size?
\end{enumerate}

\subparagraph*{\textbf{One Group of Size \(k\).}}
\begin{enumerate}[label={RQ\arabic*:}, resume]
    \item How does varying the feature dimension \(d\) affect the runtime and scalability of the algorithms in the one-group setting?
    \item How does changing the number of tuples \(n\) influence the computational performance of the algorithms in this setting?
    \item How does the group size parameter \(k\) impact the efficiency and scalability of the algorithms when there is a single protected group?
\end{enumerate}

\subparagraph*{\textbf{Real-World Dataset.}}
\begin{enumerate}[label={RQ\arabic*:}, resume]
    \item How do the algorithms behave with one group setting for a real-world dataset, and how does this configuration affect computational performance?
     \item How do the algorithms behave with a two-group setting for a real-world dataset, and how does this configuration affect computational performance compared with the one-group setting?
\end{enumerate}

\subparagraph*{\textbf{Quality of Explanations.}}

\begin{enumerate}[label={RQ\arabic*:}, resume]
    \item Across the synthetic datasets, and the real world dataset, do algorithms recover a bonus-aware linear utility function that aligns with the observed ranking? %
\end{enumerate}

\vspace{-1mm}
\subsection{Results: Performance and Scalability}

\paragraph*{\textbf{Results: RQ1 -- Singleton Groups Setting}}
We experimented with synthetic two-dimensional data, where the number of tuples given an additive bonus ($k$) was set to 10\% of the total number of tuples ($0.1n$). Figure~\ref{fig:singleton-time-d2} presents the results. As shown, within a timeout of 1800 seconds, {\singAlg} scales only up to 5{,}000 tuples, highlighting its limited practical applicability. This behavior aligns with its theoretical complexity of $\mathcal{O}(n^3 \log n)$ and motivates the need for practical MILP-based reformulations. 

For the special case of two-dimensional data, the MILP-based approach {\ilprefined} scales well, handling up to 75{,}000 tuples within 450 seconds. However, even {\ilprefined} fails to scale beyond two dimensions for singleton groups. Detailed results is shown in the appendix (Figure~\ref{fig:singleton-time-d5}), after $n = 5{,}000$ {\ilprefined} times out for $d = 5$.

For the quality analysis of the produced solutions, we compared the outputs of {\sampling}, {\singAlg}, {\ilprefined}, ordinal regression, and logistic regression whenever solutions were available. Since {\sampling} operates under a fixed time budget and relies on random sampling, the solution quality depends on the allocated runtime. As shown in Figure~\ref{fig:singleton-best}, for two-dimensional data, {\sampling} and {\ilprefined} select an almost identical number of tuples to receive an additive bonus (represented as the best value in the figure) for up to 15{,}000 tuples. However, as $n$ increases, the quality of {\sampling} begins to degrade, even for $d = 2$. Therefore, we exclude {\sampling} from the comparison in the remaining research questions. 
We observe that ordinal regression and logistic regression also find solutions of low quality, requiring a much larger number of tuples to be allocated an additive bonus to realize the observed function. This is not unexpected given that these methods are not designed to account for hidden additive bonuses. Therefore, we also exclude these baselines from the comparison in the remaining research questions.

\begin{figure*}
    \centering
    \begin{subfigure}{0.30\textwidth}
        \includegraphics[width=\linewidth,height=3.5cm]{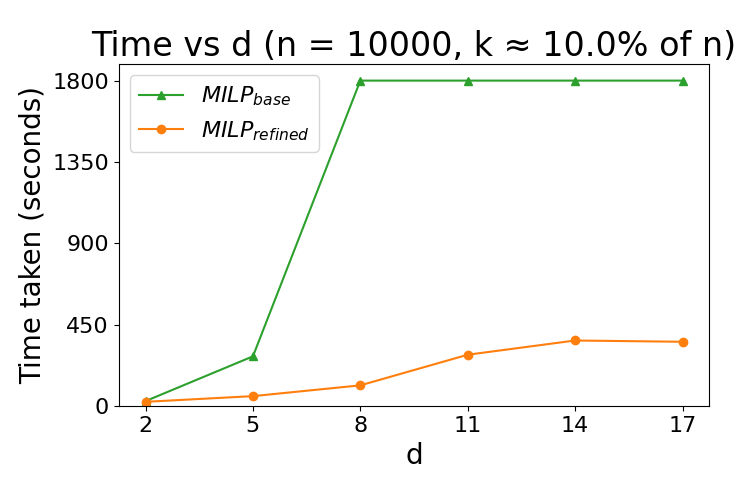}
        \vspace{-8mm}
        \caption{\footnotesize [One group setting. Impact of varying $d$.] Runtime analysis of \ilpbase and \ilprefined.}
        \label{fig:group1-time-n10000}
    \end{subfigure}
    \hfill
    \begin{subfigure}{0.30\textwidth}
          \includegraphics[width=\linewidth,height=3.5cm]{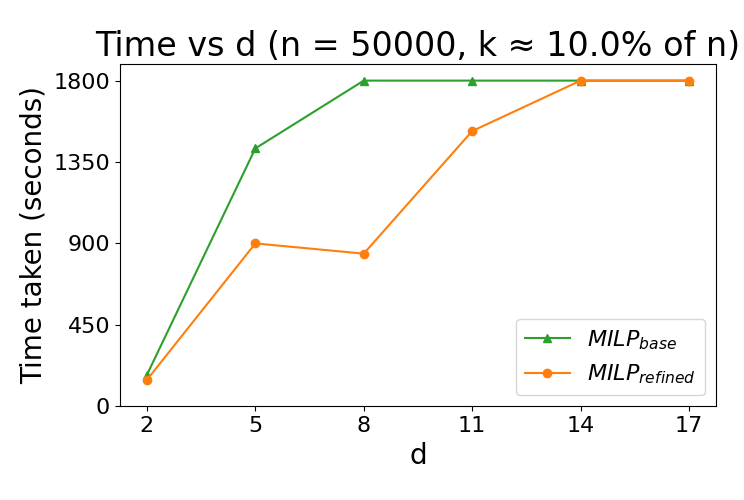}
        \vspace{-8mm}
        \caption{\footnotesize [One group setting. Impact of varying $d$.] Runtime analysis of \ilpbase and \ilprefined.}
        \label{fig:group1-time-n50000}
    \end{subfigure}
    \hfill
    \begin{subfigure}{0.35\textwidth}
          \includegraphics[width=\linewidth,height=3.5cm]{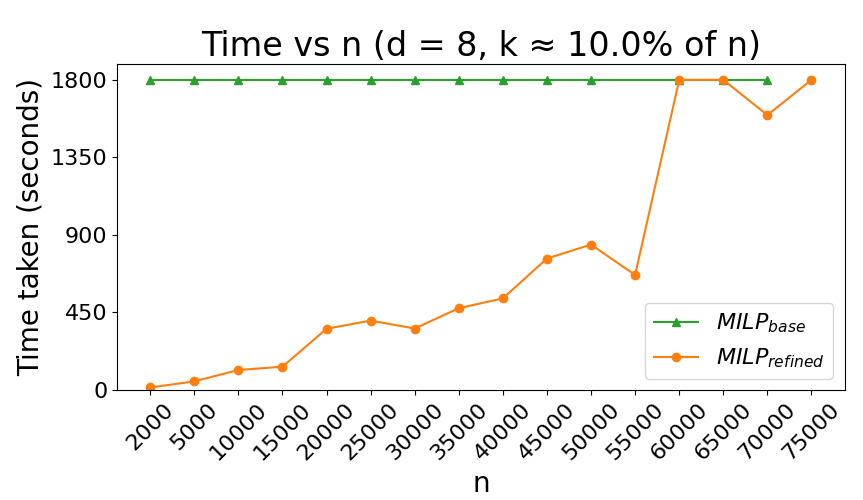}
        \vspace{-6mm}
        \caption{\footnotesize [One group setting. Impact of varying $n$.] Runtime analysis of \ilpbase and \ilprefined on $d=8$ instances.}
        \label{fig:group1-time-d8}
    \end{subfigure}
    \vspace{-3mm}
     \caption{Analysis of impact of varying $d$ (a \& b) and varying $n$ (c).}
     \vspace{-4mm}
    
\end{figure*}

\paragraph*{\textbf{Results: RQ2, RQ3 and RQ4 -- One Group Setting}}
Considering practical formulation of the MILP-based formulations, {\ilpbase} and {\ilprefined}, we now turn to RQ2--RQ4, where we study the scalability of these approaches under varying parameters. In particular, we analyze their behavior as we vary the dimensionality $d$, the number of tuples $n$, and the group size $k$.

\textbf{Regarding RQ2} -- As shown in Figures~\ref{fig:group1-time-n10000} and~\ref{fig:group1-time-n50000}, for fixed tuple sizes of $n = 10{,}000$ and $n = 50{,}000$ respectively (with $k$ set to $0.1n$), the runtime of both {\ilprefined} and {\ilpbase} increases as the dimensionality $d$ grows. For $n = 10{,}000$, the impact is particularly evident for {\ilpbase}, which times out beyond $d = 5$. In contrast, {\ilprefined} remains tractable and scales up to $d = 17$ within 450 seconds. This demonstrates that incorporating domain-specific structure significantly improves the scalability of {\ilprefined}.

However, when increasing the dataset size to $n = 50{,}000$, even {\ilprefined} times out for dimensions above $d = 11$. This indicates that while the refined formulation extends scalability, high-dimensional large instances remain computationally challenging at larger scales.
An additional scalability challenge in higher dimensions is the degradation of the skyline pruning, which is part of {\ilprefined}.  Skyline pruning is extremely effective at $d = 2$, with around 95\% of the tuples that are boosted being identified by skyline pruning. When $d=5$, the effect diminishes significantly, with around 5\% of the tuples identified, and at $d \ge 8$ the skyline pruning is unable to identify tuples that are definitely boosted.

\textbf{Regarding RQ3} -- For two-dimensional data, the increase in tuple size has only a modest impact on runtime. Although we observe a slight increase in execution time as $n$ grows, both {\ilpbase} and {\ilprefined} exhibit comparable performance in this setting. The detailed figure is presented in the appendix (Figure~\ref{fig:group1-time-d2}).

However, for $d = 8$, as shown in ~\ref{fig:group1-time-d8}, the situation changes markedly. Even at $n = 2{,}000$, {\ilpbase} already times out, clearly indicating that dimensionality has a far more significant effect on computational cost than the number of tuples. For $d = 8$, {\ilprefined} remains scalable up to $50{,}000$ tuples; beyond that point, solver-level heuristic changes introduce noticeable performance variance.

\textbf{Regarding RQ4} -- MILP solvers exhibit irregular and non-monotonic behavior: small changes in $k$ often lead to substantial fluctuations in runtime. This behavior is typical for MILPs. Although modifying $k$ only alters the tightness of the capacity constraints $\sum_i \delta_{ir} \le k$, such changes can significantly affect the strength of the MILP relaxation, the size of the feasible region for the $\delta$-variables, the degree of symmetry in the model, and ultimately the structure of the branch-and-bound tree explored by the solver. As a result, the runtime may increase or decrease unpredictably. This phenomenon is also observed in our experiments. For example, with $n = 10{,}000$ and $d = 5$, {\ilpbase} exhibits pronounced runtime peaks at intermediate values of $k$ (in the range of 10–50\%), whereas {\ilprefined} remains consistently faster while still reflecting the same underlying variability. Similar patterns are observed for $d = 8$ and larger values of $n$. Overall, these results highlight the inherent sensitivity of MILP solvers to constraint tightness. Detailed results are presented in the appendix -- Figures~\ref{fig:group1-time-d5-n10000}, \ref{fig:group1-time-d8-n10000}, and \ref{fig:group1-time-d8-n50000}.

\paragraph*{\textbf{Results: RQ5, RQ6 -- Real World Dataset}}
We now turn to the real-world dataset containing more than 300{,}000 tuples.

\begin{figure}[h]
    \centering
    \begin{minipage}{0.48\linewidth}
        \centering
        \includegraphics[width= \linewidth, height=2.9cm]{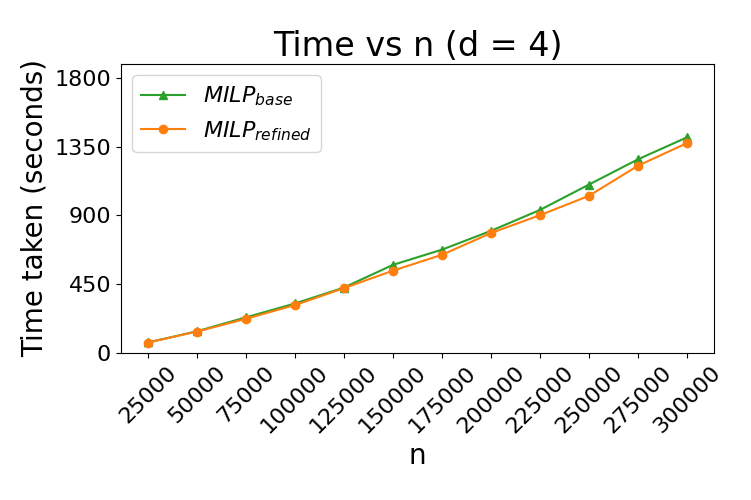}
    \vspace{-10mm}
    \caption{\footnotesize [Real world dataset. One group setting] Runtime analysis of \ilpbase and \ilprefined on instances for one group settings for JEE-dataset.}\label{fig:jee-group1}
    \vspace{-4mm}
    \end{minipage}
    \hfill
    \begin{minipage}{0.48\linewidth}
        \centering
        \includegraphics[width= \linewidth, height=2.9cm]{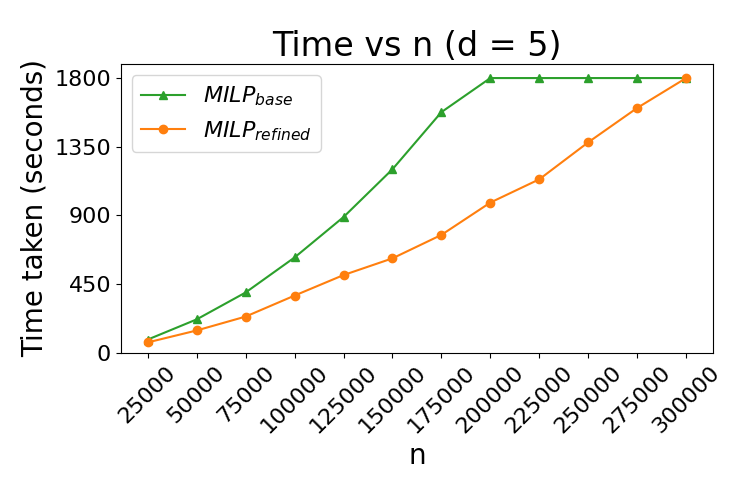}
        \vspace{-10mm}
        \caption{\footnotesize [Real world dataset. Two group setting] Runtime analysis of \ilpbase and \ilprefined on instances for two group settings for JEE-dataset.}
        \label{fig:jee-group2}
        \vspace{-4mm}
    \end{minipage}
    \vspace{-4mm}
\end{figure}

\textbf{Regarding RQ5} --- As shown in Figure~\ref{fig:jee-group1}, both {\ilpbase} and {\ilprefined} exhibit similar performance in the single-group setting. As the tuple size increases, the runtime grows approximately linearly with~$n$. Notably, we are able to find a realizable function explaining the ranking for all 300{,}000 tuples within the 1800-second timeout.

\textbf{Regarding RQ6} --- As shown in Figure~\ref{fig:jee-group2}, in the two-group setting, the domain-knowledge constraints again provide a clear advantage to {\ilprefined}. While {\ilpbase} scales up to 150{,}000 tuples, {\ilprefined} remains tractable up to 275{,}000 tuples.

\paragraph*{\textbf{Results: RQ7 -- Quality of Explanations}}
Finally, we evaluate explanation quality across all datasets. As mentioned while discussing datasets, for the synthetic datasets, both the data and the underlying linear function are known, allowing a precise assessment of reconstruction accuracy. Similarly, for the JEE dataset, the injected additive bonuses and affected tuples are also known, enabling direct evaluation of recovery accuracy. 

Across all synthetic datasets as well as the JEE dataset, the solutions produced by {\ilprefined} and {\ilpbase} exactly recover the added bonuses and the corresponding tuples. These results indicate that the MILP-based formulations scale effectively to large real-world datasets while maintaining high solution quality.

\vspace{-2mm}
\section{Related Work}
\label{sec:related}
\noindent\textbf{Utility functions}: Many classes of utility functions have been explored in prior work. \textit{Monotone utility functions}~\cite{fagin2002combining} form one such class: given two tuples $t_i$ and $t_j$ where $t_i[x] > t_j[x]$ for every attribute $x$, the function assigns a higher score to $t_i$ than to $t_j$. Efficient Top-$k$ evaluation for monotone functions was investigated by Fagin et al.~\cite{fagin2001optimal}, who proposed the Threshold and Fagin’s algorithms. Along similar lines, Xin et al.~\cite{xin2006answering} study Top-$k$ computation under user preferences over selected attributes, modeled as convex functions.

Beyond convexity, \textit{non-linear utility functions} have also been examined, including quadratic functions~\cite{kessler2015k} and more general polynomial functions~\cite{qi2018k} in multi-criteria Top-$k$ settings. Among the many families of utility functions, \textit{linear functions} have been especially influential, forming a core component of multi-criteria decision making~\cite{nanongkai2010regret}, ranking and interactive scoring systems~\cite{wang2021interactive,wang2025interactive}, and the interpretation of complex machine learning models~\cite{ribeiro2016should}.

While non-linear, monotone, and other function classes have been used for Top-$k$ generation, work on utility and scoring frameworks~\cite{shetiya2019unified,asudeh2019designing} has predominantly focused on linear functions. Following this well-established line of research, where linear functions serve as standard models for utility and scoring~\cite{shetiya2019unified,nanongkai2010regret}, we also use \textit{linear functions} as the scoring model in our work with additive bonuses.

\noindent\textbf{Learning ranking functions}: Chen et al.~\cite{chen2025synthesizing} propose \texttt{RankHow}, which explains an observed ranking using a linear function chosen to minimize Spearman’s footrule distance from $\pi$. Their formulation, however, does not allow additive bonuses, and thus may fail to realize $\pi$ when tuples must be organized into at most $g$ groups receiving group-wise boosts, as required in {\multigroup}. Fair ranking approaches~\cite{asudeh2019designing, cai2025findingfairscoringfunction} also learn linear utility functions but focus solely on satisfying group fairness constraints in the top-$k$ positions, without modeling additive adjustments.

Work on reverse top-$k$ queries~\cite{vlachou2010reverse} is also related, seeking linear functions under which specified tuples appear in the top-$k$. Why-not queries~\cite{he2014answering} similarly search for a linear function that places a designated tuple within the top-$k$, while other formulations~\cite{liu2016answering} assume a fixed set of linear functions and aim to minimally modify tuples so that they enter or exit the top-$k$. These approaches, however, do not incorporate additive bonuses and thus do not generalize to our setting; moreover, they only explain top-$k$ behavior for selected tuples and do not address the problem of explaining a complete ranking.

Traditional regression techniques -- such as linear models~\cite{maulud2020review}, decision trees~\cite{charbuty2021classification}, and neural networks~\cite{Lee2022ComplexValuedNN} -- require numeric labels and thus cannot be applied when only an ordering is available. Additionally, extension of these techniques to make them bonus-aware is not straightforward and can be an interesting work in itself. Model-agnostic explanation tools like LIME~\cite{ribeiro2016should} and its 
related IR approaches~\cite{verma2019lirme}, produce explanations by querying the underlying model through sampling of data points around the query point, which is infeasible in our setting where the ranking function is unknown
.

\emph{Our work is the first to study the problem of bonus-aware linear ranking functions as a model of explanation of a given ranking.}

\vspace{-2mm}
\section{Conclusion and Future Work}\label{sec:discussions-future-work}

In this work, we look at explanations to ranking problems based on linear scoring functions with an additive bonus model. An interesting direction is to extend or design new methods for other classes of ranking functions, such as quadratic/non-linear utility functions. Another promising open direction is to consider different models for bonus, for instance, a multiplicative bonus model, where each tuple's bonus would be computed by
$
f_w(t)=bonus(t[\tau])\cdot \left(\sum\nolimits_{j=1}^dw_j\cdot t[j]\right).
$

While our \texttt{\singAlg} approach and the hardness results apply also to the multiplicative model, our extension to \multigroup, and the ILP-based approaches do not extend. An interesting avenue of research is to design efficient algorithms for the multiplicative-bonus setting.

In this paper, we show that no efficient algorithm can exactly solve the bonus-aware linear ranking explanation problem. This hardness result suggests a natural next step: can we design algorithms that achieve a small approximation (in the number of bonus groups)? Or, under standard complexity assumptions, can we prove that even such approximations are impossible? The latter appears especially relevant given our MILP formulation, which points to a close connection with the classical \texttt{Max-$k$-LIN} problem, for which strong hardness-of-approximation results are known (e.g.,~\cite{arora1997hardness, khot2007optimal, o2011hardness, bhangale2021optimal}). This, in turn, raises the question of whether approximation algorithms for the bonus-aware linear ranking explanation problem are likewise unattainable.

\vspace{-2mm}
\section*{Acknowledgments}
{\raggedright The work of Suraj Shetiya was supported by ANRF under Grant ANRF/ECRG/2024/004976/ENS.
The work of Sujoy Bhore was supported by ANRF ARG-MATRICS, Grant 002465. 
The work of Diptarka Chakraborty was supported in part by an MoE AcRF Tier 1 grant (T1 251RES2303) and a Google South \& South-East Asia Research Award. The work of Priyanka Golia was supported by ANRF early career grant ANRF/ECRG/2024/005777/ENS.}

\bibliographystyle{ACM-Reference-Format}
\balance
\bibliography{references}

@inproceedings{chowdhury2025rankshap,
  title={RankSHAP: Shapley value based feature attributions for learning to rank},
  author={Chowdhury, Tanya and Zick, Yair and Allan, James},
  booktitle={International Conference on Learning Representations},
  volume={2025},
  pages={36765--36794},
  year={2025}
}

@inproceedings{chowdhury2023rank,
  title={Rank-lime: local model-agnostic feature attribution for learning to rank},
  author={Chowdhury, Tanya and Rahimi, Razieh and Allan, James},
  booktitle={Proceedings of the 2023 ACM SIGIR International Conference on Theory of Information Retrieval},
  pages={33--37},
  year={2023}
}

@inproceedings{heuss2025rankingshap,
  title={RankingSHAP-Faithful Listwise Feature Attribution Explanations for Ranking Models},
  author={Heuss, Maria and de Rijke, Maarten and Anand, Avishek},
  booktitle={Proceedings of the 48th International ACM SIGIR Conference on Research and Development in Information Retrieval},
  pages={381--391},
  year={2025}
}

@inproceedings{chen2025synthesizing,
  title={Synthesizing Scoring Functions for Rankings using Symbolic Gradient Descent},
  author={Chen, Zixuan and Manolios, Panagiotis and Riedewald, Mirek},
  booktitle={ICDE 2025},
  year={2025},
}

@inproceedings{asudeh2019rrr,
  title={RRR: Rank-regret representative},
  author={Asudeh, Abolfazl and Nazi, Azade and Zhang, Nan and Das, Gautam and Jagadish, HV},
  booktitle={SIGMOD 2019},
  pages={263--280},
  year={2019}
}

@inproceedings{nanongkai2012interactive,
  title={Interactive regret minimization},
  author={Nanongkai, Danupon and Lall, Ashwin and Das Sarma, Atish and Makino, Kazuhisa},
  booktitle={SIGMOD 2012},
  pages={109--120},
  year={2012}
}

@article{chen2024robust,
  title={Robust Best Point Selection under Unreliable User Feedback},
  author={Chen, Qixu and Wong, Raymond Chi-Wing},
  journal={Proceedings of the VLDB Endowment},
  volume={17},
  number={11},
  year={2024},
  publisher={VLDB Endowment}
}

@inproceedings{liu2024fair,
  title={Fair Top-k Query on Alpha-Fairness},
  author={Liu, Hao and Wong, Raymond Chi-Wing and Zhang, Zheng and Xie, Min and Tang, Bo},
  booktitle={ICDE 2024},
  pages={2338--2350},
  year={2024},
  organization={IEEE}
}

@inproceedings{wang2024reverse,
  title={Reverse regret query},
  author={Wang, Weicheng and Wong, Raymond Chi-Wing and Jagadish, HV and Xie, Min},
  booktitle={ICDE 2024},
  year={2024},
  organization={IEEE}
}

@article{fagin2002combining,
  title={Combining fuzzy information: an overview},
  author={Fagin, Ronald},
  journal={ACM SIGMOD Record},
  volume={31},
  number={2},
  pages={109--118},
  year={2002},
  publisher={ACM New York, NY, USA}
}

@article{kung1975finding,
  title={On finding the maxima of a set of vectors},
  author={Kung, Hsiang-Tsung and Luccio, Fabrizio and Preparata, Franco P},
  journal={Journal of the ACM (JACM)},
  volume={22},
  number={4},
  pages={469--476},
  year={1975},
  publisher={ACM New York, NY, USA}
}

@article{Lee2022ComplexValuedNN,
  title={Complex-Valued Neural Networks: A Comprehensive Survey},
  author={Chiyan Lee and Hideyuki Hasegawa and Shangce Gao},
  journal={IEEE/CAA Journal of Automatica Sinica},
  year={2022},
  volume={9},
  pages={1406-1426},
  url={https://api.semanticscholar.org/CorpusID:251324986}
}

@inproceedings{verma2019lirme,
  title={LIRME: locally interpretable ranking model explanation},
  author={Verma, Manisha and Ganguly, Debasis},
  booktitle={SIGIR 2019},
  pages={1281--1284},
  year={2019}
}

@article{charbuty2021classification,
  title={Classification based on decision tree algorithm for machine learning},
  author={Charbuty, Bahzad and Abdulazeez, Adnan},
  journal={Journal of applied science and technology trends},
  volume={2},
  number={01},
  pages={20--28},
  year={2021}
}

@article{maulud2020review,
  title={A review on linear regression comprehensive in machine learning},
  author={Maulud, Dastan and Abdulazeez, Adnan M},
  journal={Journal of applied science and technology trends},
  volume={1},
  number={2},
  pages={140--147},
  year={2020}
}

@article{shetiya2019unified,
  title={A unified optimization algorithm for solving" regret-minimizing representative" problems},
  author={Shetiya, Suraj and Asudeh, Abolfazl and Ahmed, Sadia and Das, Gautam},
  journal={Proceedings of the VLDB Endowment},
  volume={13},
  number={3},
  year={2019}
}

@inproceedings{wang2025interactive,
  title={Interactive Learning for Diverse Top-k Set},
  author={Wang, Weicheng and Wong, Raymond Chi-Wing and Li, Jinyang and Jagadish, HV},
  booktitle={ICDE 2025},
  year={2025},
  organization={IEEE Computer Society}
}

@article{qian2015learning,
  title={Learning user preferences by adaptive pairwise comparison},
  author={Qian, Li and Gao, Jinyang and Jagadish, HV},
  journal={Proceedings of the VLDB Endowment},
  volume={8},
  number={11},
  pages={1322--1333},
  year={2015},
  publisher={VLDB Endowment}
}

@inproceedings{wang2021interactive,
  title={Interactive search for one of the top-k},
  author={Wang, Weicheng and Wong, Raymond Chi-Wing and Xie, Min},
  booktitle={SIGMOD 2021},
  year={2021}
}

@article{nanongkai2010regret,
  title={Regret-minimizing representative databases},
  author={Nanongkai, Danupon and Sarma, Atish Das and Lall, Ashwin and Lipton, Richard J and Xu, Jun},
  journal={PVLDB 2010},
  year={2010},
  publisher={VLDB Endowment}
}

@article{qi2018k,
  title={K-regret queries using multiplicative utility functions},
  author={Qi, Jianzhong and Zuo, Fei and Samet, Hanan and Yao, Jia Cheng},
  journal={ACM TODS 2018},
  year={2018},
  publisher={ACM New York, NY, USA}
}

@article{kessler2015k,
  title={K-regret queries with nonlinear utilities},
  author={Kessler Faulkner, Taylor and Brackenbury, Will and Lall, Ashwin},
  journal={Proceedings of the VLDB Endowment},
  volume={8},
  number={13},
  year={2015},
  publisher={VLDB Endowment}
}

@inproceedings{xin2006answering,
  title={Answering top-k queries with multi-dimensional selections: The ranking cube approach},
  author={Xin, Dong and Han, Jiawei and Cheng, Hong and Li, Xiaolei},
  booktitle={VLDB},
  pages={463--475},
  year={2006},
  publisher={VLDB Endowment},
  address={Seattle, WA, USA}
}

@article{megiddo1984linear,
  title={Linear programming in linear time when the dimension is fixed},
  author={Megiddo, Nimrod},
  journal={Journal of the ACM (JACM)},
  volume={31},
  number={1},
  pages={114--127},
  year={1984},
  publisher={ACM New York, NY, USA}
}

@article{rada2018new,
  title={A new algorithm for enumeration of cells of hyperplane arrangements and a comparison with Avis and Fukuda's reverse search},
  author={Rada, Miroslav and Cerny, Michal},
  journal={SIAM Journal on Discrete Mathematics},
  volume={32},
  number={1},
  pages={455--473},
  year={2018},
  publisher={SIAM}
}

@book{edelsbrunner1987algorithms,
  title={Algorithms in combinatorial geometry},
  author={Edelsbrunner, Herbert},
  volume={10},
  year={1987},
  publisher={Springer Science \& Business Media}
}

@book{toth2017handbook,
  title={Handbook of discrete and computational geometry},
  author={Toth, Csaba D and O'Rourke, Joseph and Goodman, Jacob E},
  year={2017},
  publisher={CRC press}
}

@inproceedings{fagin2001optimal,
  title={Optimal aggregation algorithms for middleware},
  author={Fagin, Ronald and Lotem, Amnon and Naor, Moni},
  booktitle={Proceedings of the twentieth ACM SIGMOD-SIGACT-SIGART symposium on Principles of database systems},
  pages={102--113},
  year={2001}
}

@inproceedings{das2006answering,
  title={Answering top-k queries using views},
  author={Das, Gautam and Gunopulos, Dimitrios and Koudas, Nick and Tsirogiannis, Dimitris},
  booktitle={PVLDB 2006},
  pages={451--462},
  year={2006}
}

@inproceedings{chang2000onion,
  title={The onion technique: Indexing for linear optimization queries},
  author={Chang, Yuan-Chi and Bergman, Lawrence and Castelli, Vittorio and Li, Chung-Sheng and Lo, Ming-Ling and Smith, John R},
  booktitle={ACM SIGMOD 2000},
  pages={391--402},
  year={2000}
}

@article{gale2020explaining,
author = {Gale, Abraham and Marian, Am\'{e}lie},
title = {Explaining monotonic ranking functions},
year = {2020},
issue_date = {December 2020},
publisher = {VLDB Endowment},
volume = {14},
number = {4},
issn = {2150-8097},
url = {https://doi.org/10.14778/3436905.3436922},
doi = {10.14778/3436905.3436922},
journal = {Proc. VLDB Endow.},
month = dec,
pages = {640–652},
numpages = {13}
}

@inproceedings{ribeiro2016should,
  title={" Why should I trust you?" Explaining the predictions of any classifier},
  author={Ribeiro, Marco Tulio and Singh, Sameer and Guestrin, Carlos},
  booktitle={Proceedings of the 22nd ACM SIGKDD international conference on knowledge discovery and data mining},
  year={2016}
}

@article{singh2019policy,
  title={Policy learning for fairness in ranking},
  author={Singh, Ashudeep and Joachims, Thorsten},
  journal={Advances in neural information processing systems},
  volume={32},
  year={2019}
}

@inproceedings{ai2018unbiased,
  title={Unbiased learning to rank with unbiased propensity estimation},
  author={Ai, Qingyao and Bi, Keping and Luo, Cheng and Guo, Jiafeng and Croft, W Bruce},
  booktitle={SIGIR 2018},
  year={2018}
}

@article{doshi2017towards,
  title={Towards a rigorous science of interpretable machine learning},
  author={Doshi-Velez, Finale and Kim, Been},
  journal={arXiv preprint arXiv:1702.08608},
  year={2017}
}

@inproceedings{dwork2001rank,
  title={Rank aggregation methods for the web},
  author={Dwork, Cynthia and Kumar, Ravi and Naor, Moni and Sivakumar, Dandapani},
  booktitle={WWW 2001},
  pages={613--622},
  year={2001}
}

@inproceedings{guo2023rankdnn,
  title={Rankdnn: Learning to rank for few-shot learning},
  author={Guo, Qianyu and Haotong, Gong and Wei, Xujun and Fu, Yanwei and Yu, Yizhou and Zhang, Wenqiang and Ge, Weifeng},
  booktitle={Proceedings of the AAAI conference on artificial intelligence},
  volume={37},
  number={1},
  pages={728--736},
  year={2023}
}

@inproceedings{thonet2022listwise,
  title={Listwise learning to rank based on approximate rank indicators},
  author={Thonet, Thibaut and Cinar, Yagmur Gizem and Gaussier, Eric and Li, Minghan and Renders, Jean-Michel},
  booktitle={Proceedings of the AAAI Conference on Artificial Intelligence},
  volume={36},
  number={8},
  pages={8494--8502},
  year={2022}
}

@inproceedings{morik2020controlling,
  title={Controlling fairness and bias in dynamic learning-to-rank},
  author={Morik, Marco and Singh, Ashudeep and Hong, Jessica and Joachims, Thorsten},
  booktitle={SIGIR 2020},
  pages={429--438},
  year={2020}
}

@inproceedings{sen2020curious,
  title={The curious case of IR explainability: Explaining document scores within and across ranking models},
  author={Sen, Procheta and Ganguly, Debasis and Verma, Manisha and Jones, Gareth JF},
  booktitle={SIGIR 2020},
  pages={2069--2072},
  year={2020}
}

@inproceedings{mathioudakis2020affirmative,
  title={Affirmative action policies for top-k candidates selection: with an application to the design of policies for university admissions},
  author={Mathioudakis, Michael and Castillo, Carlos and Barnabo, Giorgio and Celis, Sergio},
  booktitle={ACM SAC 2020},
  year={2020}
}

@inproceedings{davenport2004computational,
  title={A computational study of the Kemeny rule for preference aggregation},
  author={Davenport, Andrew and Kalagnanam, Jayant},
  booktitle={AAAI},
  volume={4},
  pages={697--702},
  year={2004}
}

@inproceedings{rampisela2024can,
  title={Can we trust recommender system fairness evaluation? the role of fairness and relevance},
  author={Rampisela, Theresia Veronika and Ruotsalo, Tuukka and Maistro, Maria and Lioma, Christina},
  booktitle={SIGIR 2024},
  pages={271--281},
  year={2024}
}

@misc{cai2025findingfairscoringfunction,
      title={Finding a Fair Scoring Function for Top-$k$ Selection: From Hardness to Practice}, 
      author={Guangya Cai},
      year={2025},
      eprint={2503.11575},
      archivePrefix={arXiv},
      primaryClass={cs.DB},
      url={https://arxiv.org/abs/2503.11575}, 
}

@inproceedings{asudeh2019designing,
  title={Designing fair ranking schemes},
  author={Asudeh, Abolfazl and Jagadish, Hosagrahar Visvesvaraya and Stoyanovich, Julia and Das, Gautam},
  booktitle={SIGMOD 2019},
  pages={1259--1276},
  year={2019}
}

@article{liu2016answering,
  title={Answering why-not and why questions on reverse top-k queries},
  author={Liu, Qing and Gao, Yunjun and Chen, Gang and Zheng, Baihua and Zhou, Linlin},
  journal={The VLDB Journal},
  volume={25},
  number={6},
  pages={867--892},
  year={2016},
  publisher={Springer}
}

@article{he2014answering,
  title={Answering Why-Not Questions on Top-K Queries},
  author={He, Zhian and Lo, Eric},
  journal={IEEE Transactions on Knowledge and Data Engineering},
  volume={26},
  number={6},
  pages={1300--1315},
  year={2014},
  publisher={IEEE Educational Activities Department Piscataway, NJ, USA}
}

@inproceedings{vlachou2010reverse,
  title={Reverse top-k queries},
  author={Vlachou, Akrivi and Doulkeridis, Christos and Kotidis, Yannis and N{\o}rv{\aa}g, Kjetil},
  booktitle={ICDE 2010},
  pages={365--376},
  year={2010},
  organization={IEEE}
}

@inproceedings{borzsony2001skyline,
  title={The skyline operator},
  author={B{\"o}rzs{\"o}nyi, Stephan and Kossmann, Donald and Stocker, Konrad},
  booktitle={ICDE 2001},
  year={2001},
  organization={IEEE}
}

@article{arora1997hardness,
  title={The hardness of approximate optima in lattices, codes, and systems of linear equations},
  author={Arora, Sanjeev and Babai, L{\'a}szl{\'o} and Stern, Jacques and Sweedyk, Z},
  journal={Journal of Computer and System Sciences},
  volume={54},
  number={2},
  pages={317--331},
  year={1997},
  publisher={Elsevier}
}

@article{khot2007optimal,
  title={Optimal inapproximability results for MAX-CUT and other 2-variable CSPs?},
  author={Khot, Subhash and Kindler, Guy and Mossel, Elchanan and O’Donnell, Ryan},
  journal={SIAM Journal on Computing},
  volume={37},
  number={1},
  pages={319--357},
  year={2007},
  publisher={SIAM}
}

@inproceedings{shetiya2022fairness,
  title={Fairness-aware range queries for selecting unbiased data},
  author={Shetiya, Suraj and Swift, Ian P and Asudeh, Abolfazl and Das, Gautam},
  booktitle={ICDE 2022},
  year={2022},
  organization={IEEE}
}

@inproceedings{o2011hardness,
  title={Hardness of Max-2Lin and Max-3Lin over integers, reals, and large cyclic groups},
  author={O'Donnell, Ryan and Wu, Yi and Zhou, Yuan},
  booktitle={IEEE CCC 2011},
  year={2011},
  organization={IEEE}
}

@inproceedings{bhangale2021optimal,
  title={Optimal inapproximability of satisfiable k-LIN over non-abelian groups},
  author={Bhangale, Amey and Khot, Subhash},
  booktitle={Proceedings of the 53rd Annual ACM SIGACT Symposium on Theory of Computing},
  pages={1615--1628},
  year={2021}
}

@inproceedings{zhang2023finding,
  title={Finding favourite tuples on data streams with provably few comparisons},
  author={Zhang, Guangyi and Tatti, Nikolaj and Gionis, Aristides},
  booktitle={ACM SIGKDD 2023},
  year={2023}
}

@article{guo2025towards,
  title={Towards strong regret minimization sets: Balancing freshness and diversity in data selection},
  author={Guo, Hongjie and Li, Jianzhong and Gao, Hong},
  journal={Theoretical Computer Science},
  volume={1026},
  pages={114986},
  year={2025},
  publisher={Elsevier}
}

@book{schneider2022convex,
  title={Convex cones: geometry and probability},
  author={Schneider, Rolf},
  volume={2319},
  year={2022},
  publisher={Springer}
}

@inproceedings{storandt2019algorithms,
  title={Algorithms for average regret minimization},
  author={Storandt, Sabine and Funke, Stefan},
  booktitle={Proceedings of the AAAI Conference on Artificial Intelligence},
  volume={33},
  number={01},
  year={2019}
}

@inproceedings{gale2024explainable,
  title={Explainable disparity compensation for efficient fair ranking},
  author={Gale, Abraham and Marian, Am{\'e}lie},
  booktitle={ICDE 2024},
  year={2024},
  organization={IEEE}
}

@article{srinivasan1976linear,
  title={Linear programming computational procedures for ordinal regression},
  author={Srinivasan, V},
  journal={Journal of the ACM (JACM)},
  volume={23},
  number={3},
  pages={475--487},
  year={1976},
  publisher={ACM New York, NY, USA}
}

@inproceedings{10.5555/3295222.3295230,
author = {Lundberg, Scott M. and Lee, Su-In},
title = {A unified approach to interpreting model predictions},
year = {2017},
isbn = {9781510860964},
publisher = {Curran Associates Inc.},
address = {Red Hook, NY, USA},
booktitle = {Proceedings of the 31st International Conference on Neural Information Processing Systems},
pages = {4768–4777},
numpages = {10},
location = {Long Beach, California, USA},
series = {NIPS'17}
}

\appendix
\section*{Appendix}

\section{Additional experimental details}
\label{sec:app-experiment}

\begin{figure*}[htbp]
\centering
\begin{minipage}{0.32\textwidth}
    \centering
    \includegraphics[width=\linewidth,height=3.5cm]{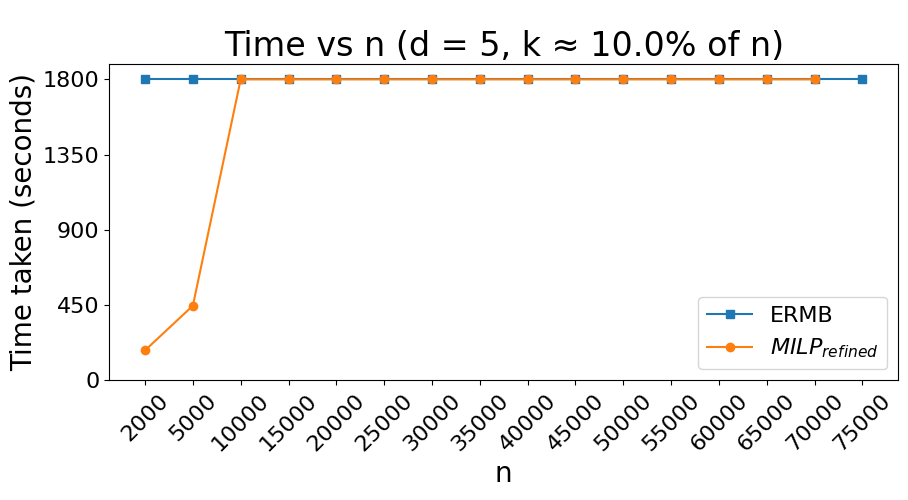}
    \vspace{-8mm}
    \caption{\footnotesize [Singleton group setting]  Runtime analysis of ERMB and \ilprefined on $d=5$ instances for singleton group settings. $k$ is fixed to 10\% of tuples $n$.}
    \label{fig:singleton-time-d5}
    \vspace{-4mm}
\end{minipage}
\hfill
\begin{minipage}{0.28\textwidth}
    \centering
    \includegraphics[width=\linewidth,height=3cm]{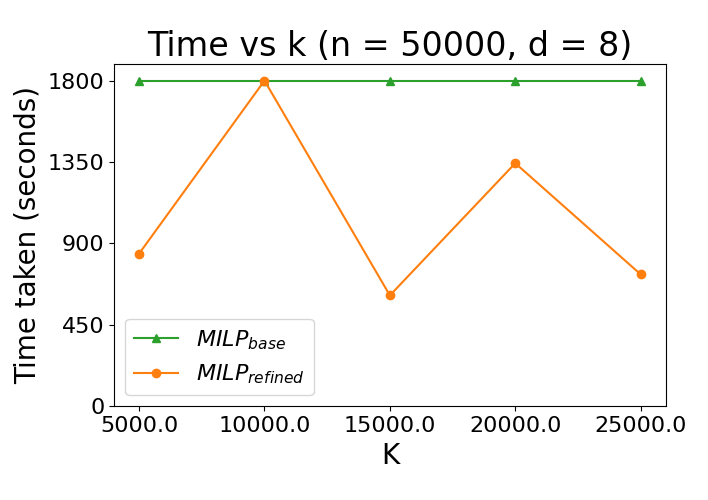}
    \vspace{-8mm}
    \caption{\footnotesize [One group setting. Analysis of the impact of varying $k$.] Runtime analysis of \ilpbase and \ilprefined on $d=8$ and $n=50000$ instances for one group settings.}
    \label{fig:group1-time-d8-n50000}
    \vspace{-4mm}
\end{minipage}
\hfill
\begin{minipage}{0.3\textwidth}
    \centering
    \includegraphics[width=\linewidth,height=3cm]{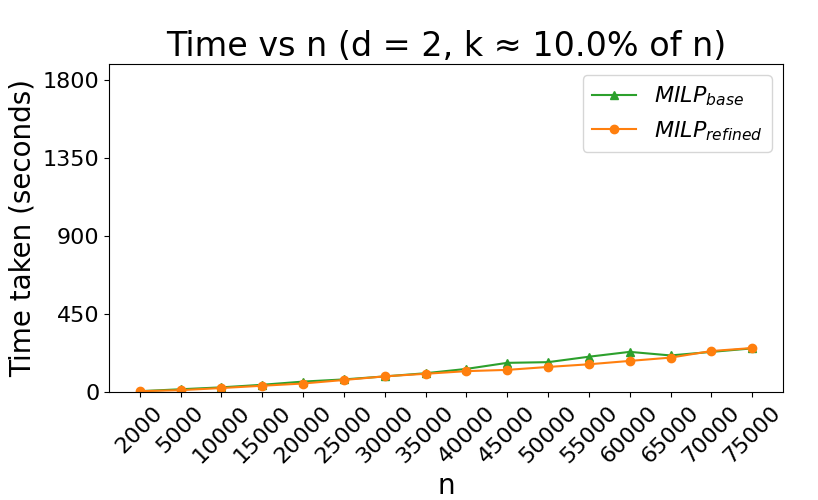}
    \vspace{-8mm}
    \caption{\footnotesize [One group setting. Analysis of the impact of varying $n$.] Runtime analysis of \ilpbase and \ilprefined on 2D instances.}
    \label{fig:group1-time-d2}
    \vspace{-4mm}
\end{minipage}
\end{figure*}

\begin{figure}[h]
    \centering
    \begin{minipage}{0.47\linewidth}
        \centering
        \includegraphics[width=\linewidth,height=3cm]{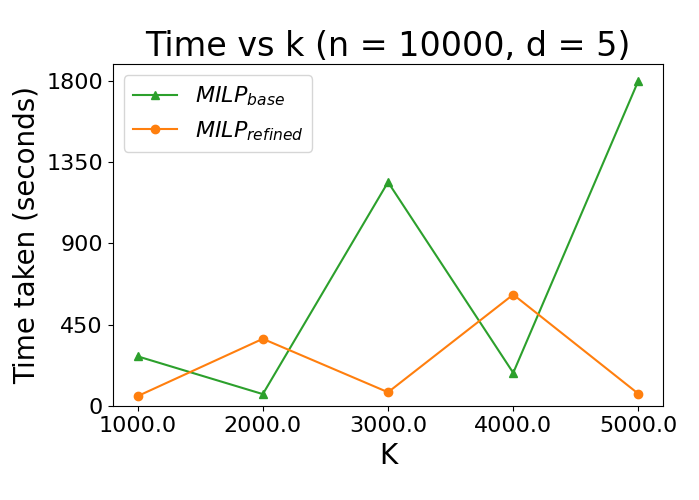}
        \caption{\footnotesize [One group setting. Analysis of the impact of varying $k$.] Runtime analysis of \ilpbase and \ilprefined on $d=5$ and $n=10000$ instances for one group settings.}
        \label{fig:group1-time-d5-n10000}
    \end{minipage}
    \hfill
    \begin{minipage}{0.47\linewidth}
        \centering
        \includegraphics[width=\linewidth,height=3cm]{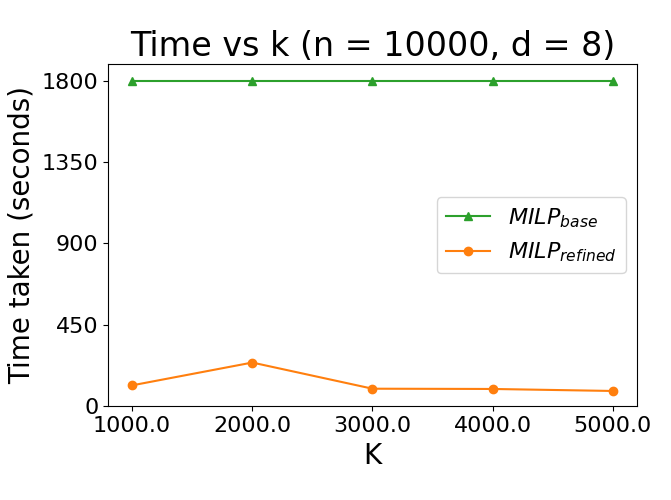}
        \vspace{-10mm}
        \caption{\footnotesize [One group setting. Analysis of the impact of varying $k$.] Runtime analysis of \ilpbase and \ilprefined on $d=8$ and $n=10000$ instances for one group settings.}
        \label{fig:group1-time-d8-n10000}
    \end{minipage}
\end{figure}

\begin{figure*}
    \centering
    \begin{subfigure}{0.30\textwidth}
        \includegraphics[width=\linewidth,height=3.5cm]{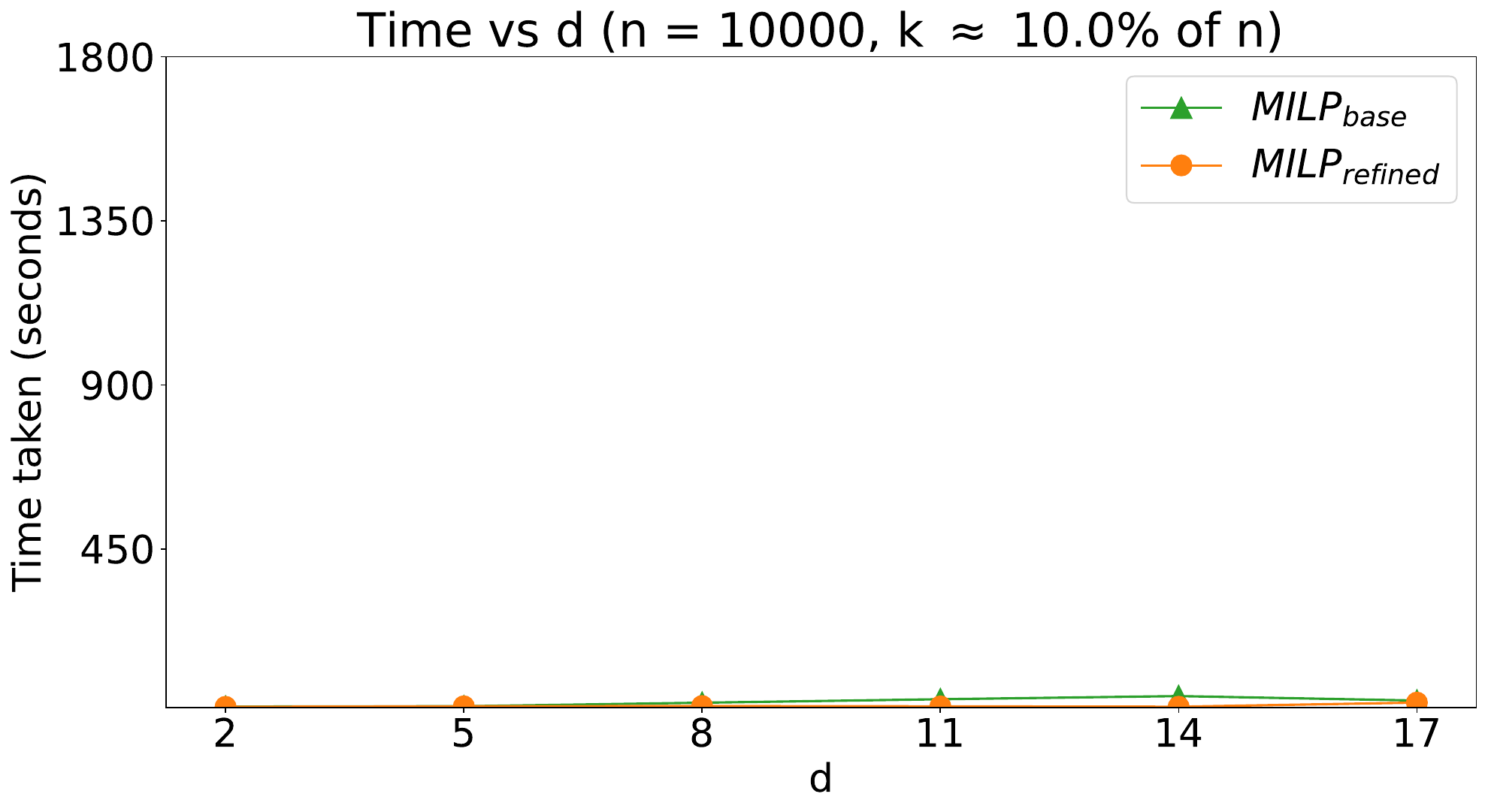}
        \caption{\footnotesize [One group setting, zipfian distribution. Impact of varying $d$.] Runtime analysis of \ilpbase and \ilprefined.}
        \label{fig:zipf-time-n10000}
    \end{subfigure}
    \hfill
    \begin{subfigure}{0.30\textwidth}
          \includegraphics[width=\linewidth,height=3.5cm]{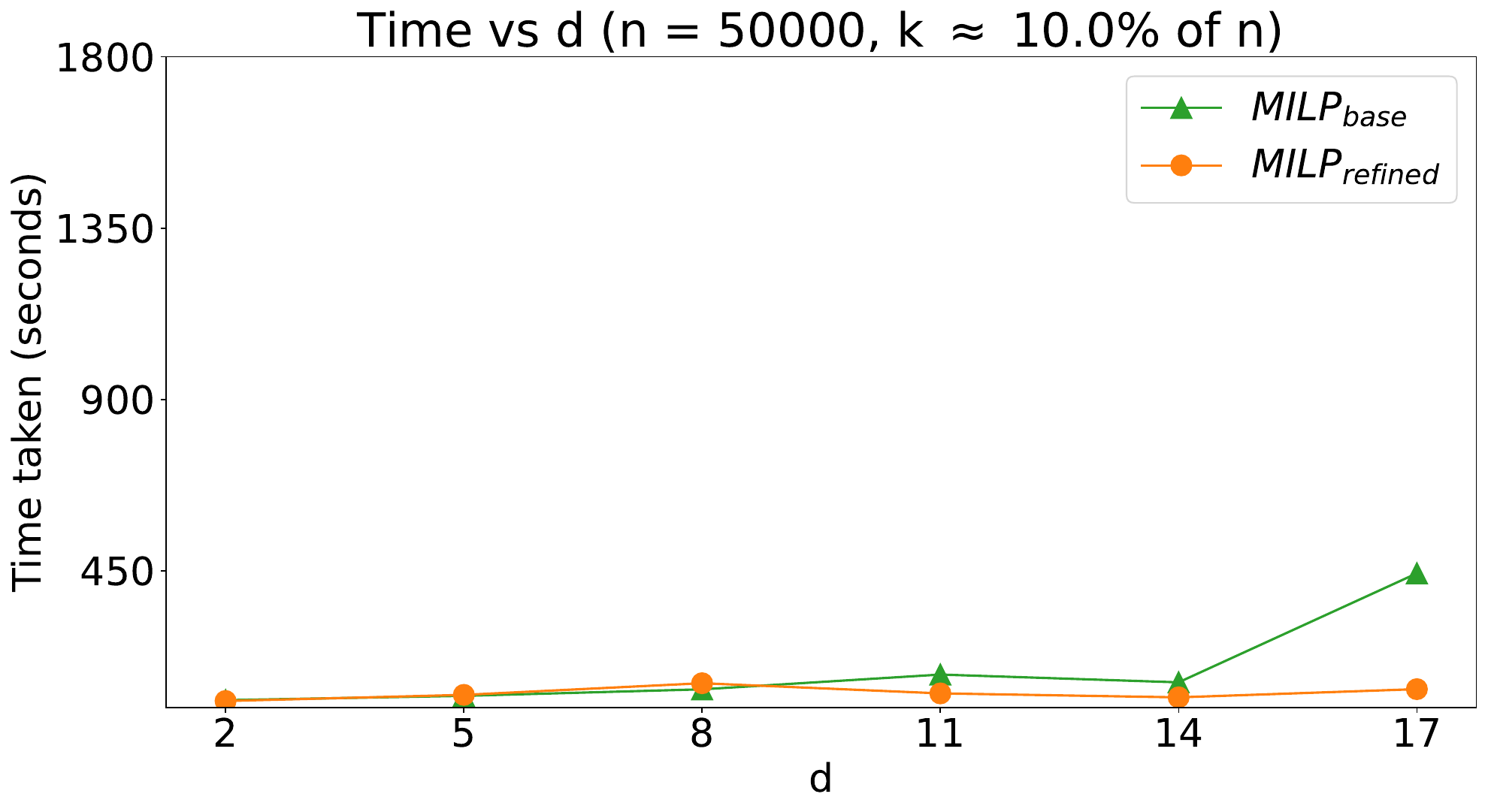}
        \caption{\footnotesize [One group setting, zipfian distribution. Impact of varying $d$.] Runtime analysis of \ilpbase and \ilprefined.}
        \label{fig:zipf-time-n50000}
    \end{subfigure}
    \hfill
    \begin{subfigure}{0.30\textwidth}
          \includegraphics[width=\linewidth,height=3.5cm]{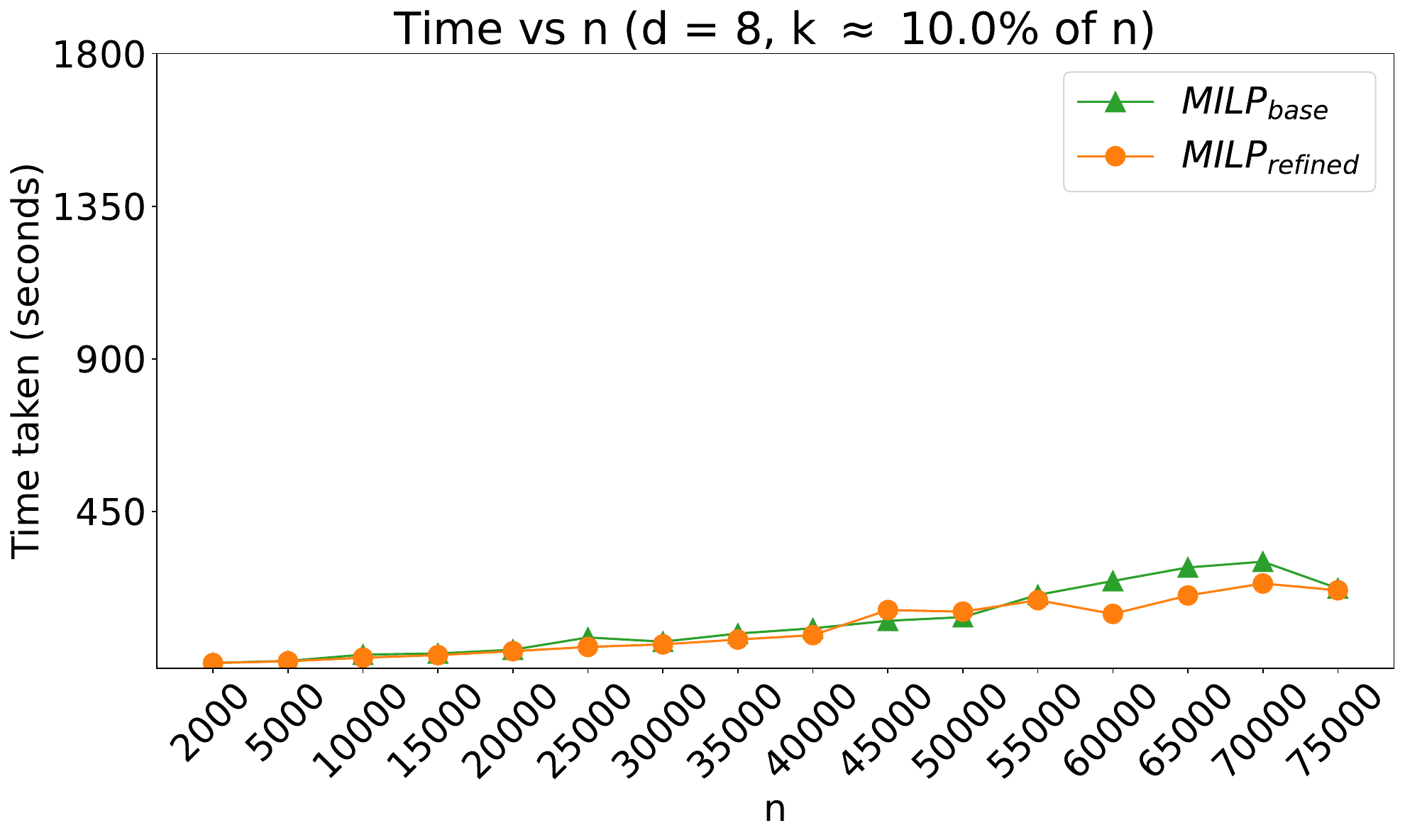}
        \caption{\footnotesize [One group setting, zipfian distribution. Impact of varying $n$.] Runtime analysis of \ilpbase and \ilprefined on $d=8$ instances.}
        \label{fig:zipf-time-d8}
    \end{subfigure}
    \vspace{-3mm}
     \caption{Analysis of impact of varying $d$ (a \& b) and varying $n$ (c).}
     \vspace{-2mm}
    
\end{figure*}

We additionally performed experiments where the tuple feature components and weight vector components are sampled from a Zeta (zipfian) distribution with parameter $2$. Figures~\ref{fig:zipf-time-n10000},~\ref{fig:zipf-time-n50000} and~\ref{fig:zipf-time-d8} show the results. The running time of both {\ilpbase} and {\ilprefined} are significantly faster on these instances as compared to when data was drawn from the uniform distribution (as shown in figures~\ref{fig:group1-time-n10000},~\ref{fig:group1-time-n50000},~\ref{fig:group1-time-d8}). The nature of the zipfian distribution increases the efficacy of the skyline pruning operation significantly, where at $d = 8$ approximately $80\%$ of the boosted tuples can be identified.

\section{Discussion on Positive Weights}
\label{sec:app-positiveweight}
In practice, there are more restrictions on the weights in the function. In many applications, each attribute either positively or negatively contributes to the overall score. In our running example, a user would be interested in obtaining the \emph{best} quality computer at the \emph{least} price. To uniformly deal with such cases, we normalize the attribute score of ``\emph{price}" (in general, any negatively contributing attribute $A$), of each item $t$ by 
\vspace{-4.5mm}
\[t[A]=\frac{\max_j(t_j[A])-t[A]}{\max_j(t_j[A])-\min_j(t_j[A])}\]

Once normalized, the weight of each attribute in the utility function can be restricted to positive values. While our techniques in the paper can handle both positive and negative weights, careful observations and optimizations based on this restriction are presented throughout the paper.

\vspace{-2mm}
\section{Missing Details of \texttt{\singAlg}}
\label{sec:app-ermb}

\begin{lemma}
\label{thm:arbitrary-add}
    \sloppy
    \texttt{\singAlg} (Algorithm~\ref{algo:singleton-groups}) has a time complexity of $\mathcal{O}(n^{2d+1} \log n)$. 
\end{lemma}

\begin{proof}
    Algorithm~\ref{algo:singleton-groups} introduces a hyperplane for every pair of tuples $t_i$, $t_j$.  As stated before, this forms a central arrangement and partitions the $d$-dimensional space into $n^{2(d-1)}$ disjoint regions. In our $(d-1)$ dimensional space, enumeration of each of the cells can be performed using the enumeration algorithm in Rada et al.~\cite{rada2018new}. \cite{rada2018new} computes a Linear Program with $d$ variables and $n^2$ constraints in every cell. Meggido~\cite{megiddo1984linear} 
    proposes a linear-time algorithm for solving an LP in fixed dimension $d$. Combining the two results, the enumeration process consumes a total of $\mathcal{O}(n^{2d})$ time.

    In every region, using the weight vector returned by the LP, the scores for the  $n$ tuples are computed in $\mathcal{O}(nd)$ time. Based on the scores, sorting is performed to obtain a ranking $\rho$. Thus, in every region, $\mathcal{O}(n\log{n})$ time is consumed to construct the ranking $\rho$. Upon computing the ranking $\rho$, LIS is computed on $\rho$ in $\mathcal{O}(n\log{n})$ time. Thus, the overall time spent in a single region is $\mathcal{O}(n\log{n})$ time.
    As there are $n^{2d-2}$ regions, our algorithm consumes a total of $\mathcal{O}(n^{2d+1} \log n)$ time.
\end{proof}

\vspace{-4mm}
\section{Details of the Algorithms for {\multigroup}}
\label{sec:app-gblr}

\begin{theorem}
\label{theorem:g-GBLR-algorithm}
    There exists a deterministic algorithm that, given a set of tuples $\mathcal{D} = \{t_1, t_2, ... t_ n \}$ where each $t_i \in \mathbb{R}^d$, a ranking $\pi$ over $\mathcal{D}$ and an integer $k$, if there exists a satisfying solution to \multigroup , returns a vector $w \in \mathbb{R}^d$, $v_1, ...,v_g \in \mathbb{R}^+$ and disjoint $G_1, ..., G_g \subseteq \mathcal{D}$ with $\sum_{i \in [g]} |G_i| \le k$ which satisfies the problem, in time $\mathcal{O} \left(g^{2(d+g)} \cdot n^{2(d+g)} \right)$.
\end{theorem}

\noindent\textbf{Embedding group bonus as new dimension}: Consider a dataset $\mathcal{D}$, ranking $\pi$ and a budget $g$ on the number of groups. Additionally, there is a budget on the sum of the sizes of the groups, $\sum_{i \in [g]} |G_i|\leq k$. 

To illustrate our idea, consider all tuples $t_j$ that belong to the group $G_i$ with a bonus value of $b_i$. If the unknown weight vector using the linear utility function was $w'=\{w_1',w_2',\ldots, w_d'\}$, then every tuple $t_j\in G_i$ would be scored as follows: $\forall_{t_j \in G_i}, \; f_{w'}(t_j)=b_i + \sum_{a=1}^d w'_a\cdot t_j[a]$.

An observation about the above equation is that $b_i$ is an additive bonus and thus for all tuples $t_j$ that belong to $G_i$, we can imagine this as a ``\emph{bonus dimension}". Therefore, for every group $G_i$, we introduce a new dimension $d+i$. All tuples $t_j \in G_i$ would have the bonus dimension set to $1$ while other tuples would have it set to $0$.  Hence, there are a total of $d+g$ dimensions in our modified setting.

But, two main things need to be handled (i) the tuples that belong to the different groups are not known beforehand, and (ii) the tuples $t_i\in \mathcal{D}$ are still in $d$-dimensional space. One observation is that we need a uniform treatment of all the tuples in the dataset. Thus, for every tuple $t_j$, we create $g+1$ copies of the tuple. The first copy corresponds to a version of the tuple $t_j$ which receives no bonus. For this version of the tuple, $t_j^{(0)}$, all bonus dimensions are set to $0$ while the original $d$ dimensions remain the same as before. For the $g$ other versions, as $t_j$ can only belong to at most one group, the version $i$ of the tuple, $t_j^{(i)}$ has the $d+i$ dimension set to $1$ and the remaining bonus dimensions set to $0$. For all versions of the tuple $t_i$ in $d+g$ dimensional space, the first $d$ dimensions are the same as those in $\mathcal{D}$. We refer to this new dataset with $(g+1)n$ tuples as $\mathcal{D}_G$.

\noindent\textbf{Multi-group version of our algorithm}: 
Our idea is to adapt \texttt{\singAlg} (Algorithm~\ref{algo:singleton-groups}) to this setting using $\mathcal{D}_G$. For every two pairs of tuples $t_i^{(x)}$, $t_j^{(y)}$ in $\mathcal{D}_G$, we introduce a \comphs\ in $d+g$ dimensional space. These $\binom{n(g+1)}{2}$ hyperplanes denoted as $\mathcal{H}_G$ produce a central arrangement which partitions the space into $(n(g+1))^{2(d+g-1)}$ regions. Similar to \texttt{ERMB}, any weight function in a region produces the same ranking of the $n(g+1)$ tuples. Thus, we enumerate each region and generate the set of all possible rankings $\mathcal{R}_G$. 

As there are $g+1$ copies of a tuple, unlike \texttt{\singAlg}, this instance of the sub-problem does not correspond to an instance of the Longest Increasing Subsequence (LIS) problem, needing us to focus on the more general Longest Common Subsequence (LCS) problem. Indeed, given two ordered lists $\pi$ and $\rho$, an LCS between $\pi$ and $\rho$ gives us the information about the groups. Thus, as a final step, we perform a \emph{budgeted} variant of LCS between $\rho \in \mathcal{R}$ and $\pi$, such that entries in $\pi$ can match any copy of the same tuple in $\rho$, while respecting the budget constraints.
The constrained version of LCS is solved by a DP with $n\times n(g+1)$ variables while maintaining a state of $n^g$. We present pseudo-code of \texttt{\multAlg} (Algorithm~\ref{algo:multi-group}) for the \multigroup\ problem and in Algorithm~\ref{alg:multi-lcs} for the budgeted variant of the LCS problem. 

Note that, similar to \texttt{\singAlg}, our \texttt{\multAlg} approach is amenable to transforming the equations from a $d$ dimensional space to a $(d-1)$ dimensional space. Our analysis makes use of this transformation in producing the runtime complexity.

\begin{algorithm}[t]
\caption{\texttt{\texttt{\multAlg}}: Multi-group Ranking Algorithm}\label{algo:multi-group}
{\small
\begin{algorithmic}[1]
\Procedure{\texttt{\multAlg}}{Dataset $\mathcal{D}$, ranking $\pi$, no. of groups $g$, size bound $k$}
\EndProcedure

\State Initialize $\mathcal{D}_G \gets \emptyset$
\ForAll{tuples $t_j \in \mathcal{D}$}
    \State Create $t_j^{(0)}$ with $d$ dims and $g$ bonus dims set to $0$
    \State Add $t_j^{(0)}$ to $\mathcal{D}_G$
    \For{$i = 1$ to $g$}
        \State Create $t_j^{(i)}$ with attribute $d+i \gets 1$, others $0$
        \State Add $t_j^{(i)}$ to $\mathcal{D}_G$
    \EndFor
\EndFor

\State Compute hyperplanes $\mathcal{H}_G$ in $\mathbb{R}^{d+g}$ from all pairs in $\mathcal{D}_G$
\State Enumerate all ranking regions $\mathcal{R}_G$ induced by $\mathcal{H}_G$
\State Initialize $\rho^* \gets \texttt{null}$,\ $\text{maxLCS} \gets 0$,\ $\mathcal{G}\gets \texttt{null}$

\ForAll{$\rho \in \mathcal{R}_G$}
    \State Compute $L$, constrained LCS  between $\rho$ and $\pi$ with size bound $k$ using \textsc{MultipleGroupLCS}
    \If{$|L| > \text{maxLCS}$}
        \State $\text{maxLCS} \gets |L|$; \ \ $\rho^* \gets \rho$;\ \ $\mathcal{G}\gets G$ from LCS output
    \EndIf
\EndFor

\State \Return linear function $s^{-1}(\rho^*)$, $\mathcal{G}$
\end{algorithmic}
}
\end{algorithm}

\begin{algorithm}[ht]
\caption{\textsc{MultipleGroupLCS}}
\label{alg:multi-lcs}
{\small
\begin{algorithmic}[1]
\Procedure{MultipleGroupLCS}{$\pi, \sigma$}
    \State Initialize table DP$[n][(g+1) \cdot n][k]$ to -1.
    \State $s \gets $ Solve($n, (g+1) \cdot n,  k)$
    \State \Return $s$
\EndProcedure

\Procedure{Solve}{$i, j, b$}
    \If {$i = 0$ or $j = 0$}
        \State \Return 0
    \EndIf
    \If {DP$[i][j][b] > -1$}
        \State \Return DP$[i][j][b]$
    \EndIf
    \If {$\pi(i) = p$ and $\sigma(j) = p^{(\ell)}$ and $b > 0$}
        \State DP$[i][j][b-1] = \max (\text{Solve}(i-1, j-1, b-1)  + 1, \text{Solve}(i, j-1, b) $
    \ElsIf {$\pi(i) = p$ and $\sigma(j) = p^{0}$}
        \State DP$[i][j][b] = 1 + \text{Solve}(i-1, j-1, b)$
    \Else
    \State DP$[i][j][b]  = \text{Solve}(i, j-1, b)$
    \EndIf
    \State \Return DP$[i][j][b]$
\EndProcedure
\end{algorithmic}}
\end{algorithm}

\subsubsection{Analysis of \texttt{\multAlg}}
We now present arguments for correctness and time complexity bounds for our algorithm.

\begin{theorem}
\label{thm:multigroup-solve}
    There exists an algorithm that finds a satisfying solution to \multigroup\ . The algorithm runs in time $\mathcal{O} \left(g^{2(d+g)+1} \cdot n^{2d + 3g + 2} \right)$.
\end{theorem}

\paragraph{Running time.}

\begin{lemma}
\label{theorem:multigroup-time}
   \sloppy
   The running time of \texttt{\multAlg} (Algorithm~\ref{algo:multi-group}) is $\mathcal{O} \left(g^{2(d+g)} \cdot n^{2(d+g)} \right)$.
\end{lemma}
\begin{proof}
Observe that for each point in $\mathcal{D}$, we create $g$ additional points in $\mathcal{D}_G$. So $|\mathcal{D}_G| = n(g+1)$. As the dimension of the weight-space has changed to $d+g$ with $(n(g+1))^2$ hyperplanes, the number of possible regions of the central arrangement is $\mathcal{O}((n(g+1))^{2(d-1+g)})$. Similar to Theorem~\ref{thm:arbitrary-add}, the regions are enumerated using \cite{rada2018new}, with the LP consuming $\mathcal{O}(n^2g^2)$ time per region, and then solving the LCS problem on the realized ranking $\rho$.

Now consider the running time of the DP to compute the LCS. It holds that $k \le n$. Therefore, the number of states in the DP is upper bounded by $\mathcal{O}(g n^2)$, and each state is solvable in $\mathcal{O}(1)$ time.

Therefore, the overall runtime of the algorithm is upper bounded by $\mathcal{O}(g^{2(d+g)} \cdot n^{2(d+g)})$.
\end{proof}

\begin{lemma}
\label{theorem:multigroup-weight}
   Given $\mathcal{D}, \pi, k$, if there exists a satisfying solution to \multigroup , then \texttt{\multAlg} (Algorithm~\ref{algo:multi-group}) returns a vector $w \in \mathbb{R}^d$, $v_1, ...,v_g \in \mathbb{R}^+$ and disjoint $G_1, ..., G_g \subseteq \mathcal{D}$ with $\sum_{i \in [g]} |G_i| \le k$ which satisfies the problem.
\end{lemma}
\begin{proof}
Suppose that there does exist some vector $w^*$, groups $G^*_1, ..., G^*_g \subset \mathcal{D}$, and additive values $v^*_1, ..., v^*_g$ that satisfy the given input. 

Consider the region in the arrangement of $\mathcal{H}$ where $w^*$ is located. Let $v^*$ be the optimal bonus scores. Consider the ranking $\sigma$ over $\mathcal{D}_G$ caused by this partition. As $w^*$ and $v^*$ must lie in one of the regions of the arrangement, an LCS of $\rho$ and $\pi$ must definitely be explored. The constrained LCS problem between $\rho$ and $\pi$ can be solved exactly. Based on this, we can obtain both the ranking function $w^*$ and the bonus scores $v^*$ along with the group information. Hence, our algorithm finds a satisfying solution if one exists.
 \end{proof}

\section{Missing Details from the Hardness Results}
\label{sec:app-NP-hard}
The following lemma establishes that a YES instance of {\maxunisat} maps to that of {\arbitrary}.

\begin{lemma}
    \label{lem:completeness}
    If there exists an assignment that satisfies at least $r$ clauses of the formula $F$ by making exactly one literal in each of these $r$ clauses true, then there exist a weight vector $w \in \mathbb{R}^n$ and at most $m - r$ singleton groups such that $\pi$ is realized by the linear utility function $s_w$ using an additive bonus of either $2$ or $-2$ to each point in those singleton groups.
\end{lemma}
\begin{proof}
    Let $\tau$ be an assignment that satisfies at least $r$ clauses of the formula $F$ by making exactly one literal in each of these $r$ clauses true. Then consider the following $w \in \mathbb{R}^n$: For each $j \in [n]$, 
    \[
w[j] = \begin{cases}
    1 \quad \quad  \text{if }  \tau \text{ sets the variable } x_j \text{ to true}\\
    -1 \quad \; \text{otherwise}.
\end{cases}
\]
    
    Note, that for each clause $C_i$ in $F$, if $\tau$ satisfies it by making exactly one literal true, then $s_w(p_i)=0$; otherwise, $s_w(p_i)$ would be either $-2$ or $2$. Further, for each $q_s$, for $s \in [\ell]$, $s_w(q_s)=0$. 
    
    Next, consider a singleton group for each point $p_i$ such that $s_w(p_i) \ne 0$. Observe that the number of such singleton groups is at most $m-r$. Furthermore, an additive bonus of either $2$ or $-2$ suffices to make their score equal to $0$, and thus with those bonuses to at most $m-r$ singleton group elements, $\pi$ is realized by the utility function $s_w$. 
\end{proof}

Next, we argue that a NO instance of {\maxunisat} maps to that of {\arbitrary}.

\begin{lemma}
    \label{lem:soundness}
    If there exist a weight vector $w \in \{\mathbb{R}\setminus \{0\}\}^n$ and at most $m - r$ singleton groups such that $\pi$ is realized by the linear utility function $s_w$ using an additive bonus to each point in those singleton groups, then there exists an assignment that satisfies at least $r$ clauses of the formula $F$ by making exactly one literal in each of these $r$ clauses true.
\end{lemma}
\begin{proof}
    Let us consider a weight $w \in \{\mathbb{R}\setminus \{0\}\}^n$ such that with additive bonuses to the minimum number of singleton groups, $\pi$ is realized by $s_w$. Let $\mathcal{G}$ be the set of such singleton groups. By our assumption in the lemma, $|\mathcal{G}| \le m-r$. We would like to claim that any such singleton group (i.e., $\in \mathcal{G}$) must contain $p_i$, for some $i \in [m]$, but not any $q_j$'s.

    To argue that, for the sake of contradiction, let there exist $G \in \mathcal{G}$ such that $G=\{q_j\mid \text{for some }j \in [\ell]\}$. Let $j = j_1 n^2 + j_2$, for some $j_1 \in [m+1]$, $j_2 \in [n^2]$. Then there must exist $n^2$ many groups in $\mathcal{G}$ -- a group for each point $q_{j_1 n^2 + t}$, for each $t \in [n^2]$. Otherwise, if any of these groups is not present in $\mathcal{G}$, one can exclude $G=\{q_j\}$ from the set of singleton bonus groups $\mathcal{G}$ and still $\pi$ is realized by $s_w$ (this is because for all $q_i$, $s_w(q_i)=0$), thus violating the minimality of size of $\mathcal{G}$. Now, since $|\mathcal{G}| \le m-r$, it cannot contain any singleton group with point $q_j$, for some $j \in [\ell]$.

    Next, we claim that if for some $p_i$, $\{p_i\} \not \in \mathcal{G}$, we must have that $s_w(p_i)=0$. Otherwise, by an argument similar to the above, we can show that $\mathcal{G}$ must contain $n^2$ many singleton groups, thus violating the assumption $|\mathcal{G}| \le m-r$. 
    
    Now, consider any $p_i$, such that $\{p_i\} \not \in \mathcal{G}$. By the construction, $p_i$ has two non-zero coordinates, say $a,b \in [n]$. Then $w_a = - w_b$. Next, consider the corresponding clause $C_i$. Observe, $C_i$ would be satisfied by making exactly one literal of it true, by assigning $x_a = true$ and $x_b =  false$, or $x_a = false$ and $x_b =  true$. 

    Now, construct an assignment $\tau$ as follows: It sets $x_a = true$ if $w_a > 0$; otherwise $x_a = false$. Also, consider another assignment $\bar{\tau}$, which is obtained by flipping the assignment of each variable in $\tau$. By the argument in the previous paragraph, observe that either $\tau$ or $\bar{\tau}$ satisfies all the clauses $C_i$ for which $\{p_i\} \not \in \mathcal{G}$, by making exactly one literal of each of them to true, and this concludes the proof of the lemma.
\end{proof}

\begin{proof}[Proof of Theorem~\ref{thm:np-hard}]
   It is straightforward to see that the reduction from the {\maxunisat} to {\arbitrary}, which we have provided, is polynomial time. The correctness of the reduction immediately follows from Lemma~\ref{lem:completeness} and Lemma~\ref{lem:soundness}, and that completes the proof of Theorem~\ref{thm:np-hard}.
\end{proof}

\end{document}